\documentclass{article}
\pdfoutput=1
\usepackage{euscript,amsmath,amssymb,amsfonts,graphicx,bm,amsthm}
\usepackage{cite}

  \usepackage{paralist}
  \usepackage{graphics} 
\usepackage[caption=false]{subfig}
\usepackage{soul}

\usepackage{latexsym,amssymb,enumerate,amsmath,amsthm} 
  \usepackage{graphicx} 
  \usepackage{tikz}
     \usepackage[colorlinks=false]{hyperref}
 \hypersetup{urlcolor=blue, citecolor=red}
\usepackage{latexsym,amssymb,enumerate,amsmath} 
\usepackage{pgfplots}
\usepackage{soul,xcolor}

\usepackage{mathtools}

\newcommand{\dd}{\textup{d}}
\def\eps{\varepsilon}
\def\E{\mathbb{E}}
\def\P{\mathbb{P}}
\def\R{\mathbb{R}}

\def\target{\textup{target}}

\newtheorem{theorem}{Theorem}
\newtheorem{proposition}[theorem]{Proposition}
\newtheorem{lemma}[theorem]{Lemma}
\newtheorem{corollary}[theorem]{Corollary}
\theoremstyle{plain}
\theoremstyle{remark}

\theoremstyle{definition}
\newtheorem*{definition*}{Definition}

\begin{document}


\title{Slowest first passage times, redundancy, and menopause timing}


\author{Sean D. Lawley\thanks{Department of Mathematics, University of Utah, Salt Lake City, UT 84112 USA (\texttt{lawley@math.utah.edu}). SDL was supported by the National Science Foundation (Grant Nos.\ NSF CAREER DMS-1944574 and NSF DMS-1814832).} \and Joshua Johnson\thanks{University of Colorado-Anschutz Medical Campus, Department of Obstetrics and Gynecology, Division of Reproductive Endocrinology and Infertility, Aurora, CO, USA (\texttt{joshua.2.johnson@cuanschutz.edu}). JJ was supported by University of Colorado School of Medicine Research Funds and McPherson Family Funds.} \thanks{University of Colorado-Anschutz Medical Campus, Department of Obstetrics and Gynecology, Division of Reproductive Sciences, Aurora, CO, USA}
}
\date{\today}
\maketitle

\begin{abstract} 
Biological events are often initiated when a random ``searcher'' finds a ``target,'' which is called a first passage time (FPT). In some biological systems involving multiple searchers, an important timescale is the time it takes the slowest searcher(s) to find a target. For example, of the hundreds of thousands of primordial follicles in a woman's ovarian reserve, it is the slowest to leave that trigger the onset of menopause. Such slowest FPTs may also contribute to the reliability of cell signaling pathways and influence the ability of a cell to locate an external stimulus. In this paper, we use extreme value theory and asymptotic analysis to obtain rigorous approximations to the full probability distribution and moments of slowest FPTs. Though the results are proven in the limit of many searchers, numerical simulations reveal that the approximations are accurate for any number of searchers in typical scenarios of interest. We apply these general mathematical results to models of ovarian aging and menopause timing, which reveals the role of slowest FPTs for understanding redundancy in biological systems. We also apply the theory to several popular models of stochastic search, including search by diffusive, subdiffusive, and mortal searchers.
\end{abstract}

\section{Introduction}

Timescales in many biological systems have been studied using first passage times (FPTs) \cite{chou2014, polizzi2016}. Generically, a FPT is the first time a random ``searcher'' finds a ``target.'' Depending on the application, the searcher could be, for example, an ion, protein, cell, or predatory animal, and the target could be a receptor, ligand, cell, or prey. Many mathematical and numerical methods have been developed in order to estimate such FPTs \cite{redner2001, benichou2010, cheviakov2010, opplestrup2006, kaye2020}. In the past several decades, FPT analysis has focused almost exclusively on the distribution and statistics of a single given searcher. 

There has recently been a surge of interest in the fastest FPT, which is the time it takes the fastest searcher to find a target out of multiple searchers \cite{meerson2015, godec2016x,hartich2018,hartich2019, basnayake2019, schuss2019, lawley2020esp1, lawley2020uni, lawley2020dist}. To describe more precisely, suppose there are $N\ge1$ searchers and let $\tau_{1},\dots,\tau_{N}$ denote their respective FPTs to some target. The fastest FPT is then
\begin{align}\label{t1n}
T_{1,N}
:=\min\{\tau_{1},\dots,\tau_{N}\}.
\end{align}
If $N$ is large, then the fastest FPT is much faster than a typical single FPT,
\begin{align}\label{faster}
T_{1,N}
\ll\tau,\quad\text{if }N\gg1,
\end{align}
and previous work has studied the decay of $T_{1,N}$ in this many searcher limit. We note that fastest FPTs are often called extreme FPTs, since $T_{1,N}$ in \eqref{t1n} is an example of an extreme statistic \cite{coles2001, falkbook, haanbook}. The theory of extreme statistics has been used for many decades in fields such as engineering, earth sciences, and finance \cite{coles2001, novak2011}, but the theory has only started to be applied in biology.

Much of the interest in fastest FPTs stems from attempts to understand biological ``redundancy'' \cite{meerson2015}. A prototypical example of such redundancy occurs in human reproduction, in which roughly $N=10^{8}$ sperm cells search for an oocyte despite the fact that only one sperm cell initiates fertilization \cite{eisenbach2006} {(human fertilization also involves many other complicated mechanisms \cite{fitzpatrick2020})}. Other important examples come from (i) gene regulation, in which only the fastest few of the $N\in[10^{2},10^{4}]$ transcription factors determine cellular response \cite{harbison2004}, and (ii) intracellular calcium dynamics, in which the fastest two of the $N=10^{3}$ released calcium ions to arrive at a Ryanodyne receptor trigger further calcium release \cite{basnayake2019fast}. In these systems, why are there $N\gg1$ searchers when only a few searchers determine the biological response? Since the fastest FPT is much faster than a typical single FPT (as in \eqref{faster}), it has been argued that the apparently redundant or ``extra'' searchers in these systems are not a waste, but rather function to accelerate the search process \cite{schuss2019, coombs2019, redner2019, sokolov2019, rusakov2019, martyushev2019, tamm2019, basnayake2019c}. 

Rather than the fastest FPT, an important timescale in some biological systems is the time it takes the slowest searcher(s) to find a target, which we call a slowest FPT. A slowest FPT can define the termination of a process or perhaps the exhaustion of a supply. More precisely, let us generalize \eqref{t1n} and define the order statistics,
\begin{align*}
T_{1,N}\le T_{2,N}\le \cdots\le T_{N-1,N}\le T_{N,N},
\end{align*}
where $T_{j,N}$ denotes the $j$th fastest FPT,
\begin{align}\label{Tkn}
T_{j,N}
:=\min\big\{\{\tau_{1},\dots,\tau_{N}\}\backslash\cup_{i=1}^{j-1}\{T_{i,N}\}\big\},\quad j\in\{1,\dots,N\}.
\end{align}
In this notation, a key quantity in some systems is the $(k+1)$st slowest FPT,
\begin{align*}
T_{N-k,N},\quad\text{for }{{N\gg k+1\ge1}}.
\end{align*}

For example, the onset of menopause is triggered by a slowest FPT. When a woman is born, she has around $N\approx5\times10^{5}$ primordial ovarian follicles (PFs) in her ovarian reserve \cite{wallace2010}. During her lifetime, this number {decays} as individual PFs in this dormant reserve enter a growth stage (given that no new PFs are formed after birth). Menopause occurs when the number of PFs in the reserve drops to around $k\approx10^{3}$, which is around age 50 years for most women \cite{faddy1992}. Hence, if  $\{\tau_{n}\}_{n=1}^{N}$ denote the times that each of the PFs {leaves the reserve}, then menopause occurs at time $T_{N-k,N}$, where
\begin{align}\label{phys}
N\approx5\times10^{5}\gg k\approx10^{3}.
\end{align}
Put another way, the timing of menopause is determined by the slowest $k/N\approx0.2\%$ of PFs to {leave the reserve}.

Furthermore, this ovarian system is notable for its apparent redundancy. Indeed, of the hundreds of thousands of PFs present at birth, most are destined to die some time after entering the growth stage, and only about one PF survives to ovulate per menstrual cycle \cite{faddy1992, faddy1996}. Over 40 years of menstrual cycles, only approximately $40\times12\approx500$ PFs are relevant to possible reproduction across the lifetime. Even considering the larger number of follicles that engage in ovarian endocrine function and participate in the signaling required for continued menstrual cyclicity (and then die), why have $N\approx5\times10^{5}$ PFs when only much smaller fractions are absolutely needed? What explains this disparity of up to three orders of magnitude? The situation is exacerbated by the fact that at around 5 months of gestation, a female has closer to $5\times10^{6}$ PFs in her reserve \cite{wallace2010}. To begin to address this seeming redundancy in the ovarian system, a mathematical understanding of slowest FPTs is required.

Slowest FPTs also play a role in cell signaling pathways. A prototypical model involves signaling molecules (i.e.\ searchers) diffusing from a source and then binding to some receptor (i.e.\ the target) \cite{liu2018}. The very interesting recent study by \cite{ma2020} found that such cell signaling mechanisms are strongly affected by signal inactivation, in which the diffusing searchers can be inactivated before finding the target (such searchers are often called ``mortal'' or ``evanescent'' \cite{abad2010, abad2012, abad2013, yuste2013, meerson2015b, grebenkov2017}). In particular, \cite{ma2020} showed that inactivation ``sharpens'' signals by reducing variability in the FPTs of a continuum of many searchers arriving at the target. For a signal conveyed by a finite number of searchers, signal sharpness or ``spread'' could be understood in terms of the difference between the latest and earliest searchers to arrive at the target. That is, the spread of a signal relayed by the arrival of discrete searchers could be defined as
\begin{align}\label{sharpness}
\Sigma_{N}
:=T_{N,N}-T_{1,N}.
\end{align}
How does this notion of signal spread depend on inactivation? How does $\Sigma_{N}$ depend on cellular geometry and the many other parameters in the problem? Answering these questions requires understanding slowest FPTs $T_{N,N}$ and how they relate to fastest FPTs $T_{1,N}$.

Understanding slowest FPTs also promises insight into single-cell source location detection. Many types of cells display a remarkable ability to pinpoint the location of an external stimulus. Examples include eukaryotic gradient-directed cell migration (chemotaxis) \cite{levchenko2002}, directional growth (chemotropism) in growing neurons \cite{goodhill2016}, and yeast \cite{ismael2016}. In these systems, cells infer the spatial location of an external source through the noisy arrivals of diffusing molecules (searchers) to membrane receptors (targets). Dating back to the seminal work of \cite{berg1977}, there is a long history of mathematical modeling of such systems \cite{zwanzig1990, zwanzig1991, bernoff2018b, lindsay2017,Berezhkovskii2004, Berezhkovskii2006, Dagdug2016, Eun2017, Muratov2008, lawley2019bp, eun2020, handy2021}. 
{A recent} work investigated the theoretical limits of what a cell could infer about the source location from the number of arrivals at different membrane receptors \cite{lawley2020prl}. Intuitively, membrane receptors which receive many molecules are likely nearer the source than receptors which receive only a few molecules. However, this prior study considered only the total number of arrivals at each receptor, rather than the temporal data of when the molecules arrive at each receptor. We conjecture that including the arrival time data would significantly improve the estimates of the source location, and might thereby improve the rather inaccurate estimates found in \cite{lawley2020prl} for sources located more than a few cell radii away. Indeed, recent work shows that early arrivals are much more likely to hit receptors near the source compared to later arrivals \cite{linn2021}. However, quantifying how arrival time data helps pinpoint the source location again requires understanding slowest FPTs.

We also note that slowest FPTs have recently been studied in the chemical physics literature in the very interesting work by \cite{grebenkov2022}. {These authors study slowest FPTs as a special case of so-called impatient particles problem, which involves particles which bind reversibly to a target. {To our knowledge, the only prior study of slowest FPTs is the recent work by \cite{grebenkov2022}.}} See section~\ref{sg} for a description of how the present paper yields (i) an exact and mathematically rigorous result that explains a result by \cite{grebenkov2022} which was obtained therein by numerical fitting and (ii) higher order corrections to an estimate by \cite{grebenkov2022}.

In this paper, we use extreme value theory and asymptotic analysis to obtain rigorous approximations to the full probability distribution of $T_{N-k,N}$ in the many searcher limit, ${{N\gg k+1\ge1}}$. We further prove asymptotic expansions of all the moments of $T_{N-k,N}$ in this limit. Though the results are proven in the limit of many searchers, numerical simulations reveal that the approximations are accurate for any number of searchers in typical scenarios of interest. This contrasts existing estimates of fastest FPTs, which require very large values of $N$ to be accurate. 

In line with previous work on fastest FPTs, we assume that the single FPTs $\tau_{1},\dots,\tau_{N}$ are independent and identically distributed (iid). Our results are given in terms of the large-time distribution of a single FPT $\tau$. Depending on this large-time distribution, our results involve rescaling $T_{N-k,N}$ according to a diverging power function, logarithm, harmonic number, or {Lambert W} function of the searcher number $N$. {We emphasize that these general mathematical results apply to the largest values $T_{N-k,N},\dots,T_{N,N}$ of any sequence of iid nonnegative random variables $\{\tau_{n}\}_{n\ge1}$ whose large-time distribution decays either (a) exponentially (possibly with a power law pre-factor) or (b) according to a power law.}

We apply these general mathematical results to {models} of ovarian aging and menopause timing, which reveals the role of slowest FPTs for understanding redundancy in biological systems. We also apply the   mathematical results to several common models of stochastic search in biology. We consider diffusive searchers, subdiffusive searchers, and searchers which move on a discrete network. We also consider searchers which can be inactivated before finding the target (i.e.\ mortal searchers). We find that even small inactivation rates can drastically sharpen the signal by decreasing the signal spread, $\Sigma_{N}$, in \eqref{sharpness}.

The rest of the paper is organized as follows. In section~\ref{math}, we present the general mathematical results. In section~\ref{meno}, we apply the results to a model of ovarian aging and menopause timing. In section~\ref{examples}, we apply the results to various search models and compare the asymptotic theory to numerical simulations. In section~\ref{mortal}, we consider mortal searchers. {Sections~\ref{examples} and \ref{mortal} reveal several generic features of slowest FPTs, which offers insight into single-cell source location detection and cellular signaling.} In the Discussion section, we discuss further applications of the theory and describe how slowest FPTs can be nearly deterministic. We collect the mathematical proofs and some technical details in an appendix.

\section{Mathematical analysis}\label{math}

In this section, we present general mathematical results on slowest FPTs. The proofs are given in section~\ref{appproofs} in the Appendix. The results are formulated in terms of the large-time asymptotic behavior of the so-called survival probability of a single searcher, which we denote by
\begin{align*}
S(t)
:=\P(\tau>t).
\end{align*}
As in the Introduction, $\{\tau_{n}\}_{n\ge1}$ is an iid sequence of realizations of a nonnegative random variable $\tau>0$. The order statistics,
\begin{align*}
T_{1,N}\le T_{2,N}\le\cdots\le T_{N,N},
\end{align*}
are defined in \eqref{Tkn}. {Though we consider slowest FPTs in sections~\ref{meno}-\ref{mortal}, we emphasize that the results of the present section are general results which apply to any iid sequence of nonnegative random variables $\{\tau_{n}\}_{n\ge1}$ whose survival probabilities decay at large time either (a) exponentially (possibly with a power law pre-factor) or (b) according to a power law.}

Since $\{\tau_{n}\}_{n\ge1}$ are iid, the distribution of $T_{N-k,N}$ for $k\in\{0,1,\dots,N-1\}$ is
\begin{align}\label{unwieldy}
\begin{split}
\P(T_{N-k,N}\le t)
&=\sum_{i=0}^{k}{N\choose i}(1-S(t))^{N-i}(S(t))^{i}\\
&=1-\P(T_{N-k,N}> t)\\
&=1-\sum_{i=0}^{N-k-1}{N\choose i}(S(t))^{N-i}(1-S(t))^{i}.
\end{split}
\end{align}
Though this expression gives the exact distribution of $T_{N-k,N}$, it requires knowing the survival probability $S(t)$ for all $t\ge0$ (i.e.\ it requires knowledge of the full distribution of $\tau$). In this section, we focus on obtaining approximations to the distribution and moments of $T_{N-k,N}$ in the limit $N\to\infty$ with $k\ge0$ fixed, assuming only knowledge of the large-time decay of the survival probability $S(t)$. 

Throughout this paper, 
\begin{align*}
&\text{``$f\sim g$''}\quad\text{denotes}\quad f/g\to1,\\
&\text{``$f=o(g)$''}\quad\text{denotes}\quad f/g\to0,\\
&\text{``$f=O(g)$''}\quad\text{denotes}\quad \limsup |f|/g<\infty.
\end{align*} 
{Recall} that a sequence of random variables $\{X_{N}\}_{N\ge1}$ is said to \emph{converge in distribution} to a random variable $X$ as $N\to\infty$ if
\begin{align}\label{cddef}
\P(X_{N}\le x)
\to\P(X\le x)\quad\text{as }N\to\infty
\end{align}
for all points $x\in\R$ such that the function $F(x):=\P(X\le x)$ is continuous. If \eqref{cddef} holds, then we write
\begin{align*}
X_{N}
\to_{\dd}X\quad\text{as }N\to\infty.
\end{align*}

We consider the case that $S(t)$ decays exponentially at large-time (possibly with a power law pre-factor) in sections~\ref{exponentialdecay}-\ref{bestdecay} and consider a power law decay of $S(t)$ in section~\ref{powerlawdecay}.

\subsection{Exponential decay}\label{exponentialdecay}

The first theorem yields the limiting distribution of the slowest FPT assuming that the survival probability decays exponentially at large time. The proof uses classical extreme value theory and properties of the {Lambert W} function \cite{corless1996, zarfaty2021}.

\begin{theorem}\label{expdist}
Assume
\begin{align*}
S(t)
\sim A(\lambda t)^{-p}e^{-\lambda t}\quad\text{as }t\to\infty,
\end{align*}
for some $A>0$, $\lambda>0$, and $p\in\R$. Then for any fixed integer $k\ge0$, we have 
\begin{align}\label{cdt}
\lambda T_{N-k,N}-b_{N}
\to_{\dd}Y_{k}\quad\text{as }N\to\infty,
\end{align}
where $Y_{k}$ has distribution
\begin{align*}
\P(Y_{k}\le {y})
=\frac{\Gamma(k+1,e^{-{y}})}{k!},\quad {y}\in\R,
\end{align*}
where $\Gamma(a,z):=\int_{z}^{\infty}u^{a-1}e^{-u}\,\dd u$ denotes the upper incomplete gamma function, and $\{b_{N}\}_{N\ge1}$ is any sequence satisfying
\begin{align}\label{bnp}
\lim_{N\to\infty}(b_{N}-b_{N}')=0,
\end{align}
where
\begin{align}\label{bN}
b_{N}'
:=\begin{cases}
\ln(AN) & \text{if }p=0,\\
pW_{0}(\frac{1}{p}(AN)^{1/p}) & \text{if }p>0,\\
pW_{-1}(\frac{1}{p}(AN)^{1/p}) & \text{if }p<0,
\end{cases}
\end{align}
where $W_{0}(z)$ denotes the principal branch of the {Lambert W} function and $W_{-1}(z)$ denotes the lower branch \cite{corless1996}. 
\end{theorem}

One choice of the sequence $\{b_{N}\}_{N\ge1}$ in Theorem~\ref{expdist} is $b_{N}=b_{N}'$ in \eqref{bN}. Another choice which satisfies \eqref{bnp} is 
\begin{align}\label{bNalt}
b_{N}
=\ln(AN)
-p\ln\Big|\ln\big|\frac{1}{p}(AN)^{1/p}\big|\Big|
-p\ln\big|p\big|,\quad p\in\R,
\end{align}
where the two terms involving $p$ in \eqref{bNalt} are interpreted as zero if $p=0$. 
The fact that Theorem~\ref{expdist} holds with $\{b_{N}\}_{N\ge1}$ in \eqref{bNalt} follows from the following logarithmic expansions of the {Lambert W} function \cite{corless1996},
\begin{align*}
W_{0}(z)
&=\ln z-\ln\ln z+o(1),\quad\text{as }z\to\infty,\\
W_{-1}(z)
&=\ln(-z)-\ln(-\ln(-z))+o(1),\quad\text{as }z\to0-.
\end{align*}

Roughly speaking, the convergence in distribution in \eqref{cdt} in Theorem~\ref{expdist} means
\begin{align*}
\P(T_{N-k,N}\le t)
\approx
\P(Y_{k}\le \lambda t+b_{N})
=\frac{\Gamma(k+1,e^{b_{N}}e^{-{\lambda t}})}{k!},\quad\text{if }{{N\gg k+1\ge1}}.
\end{align*}
In the case that $k=0$, the limiting distribution is Gumbel, $\P(Y_{0}\le {y})=\exp(-e^{-{y}})$, and so
\begin{align*}
\P(T_{N,N}\le t)
\approx
\P(Y_{0}\le \lambda t+b_{N})
=\exp(-e^{b_{N}}e^{-{\lambda t}}),\quad\text{if }N\gg 1.
\end{align*}

We further note that $Y_{k}$ in {Theorem~\ref{expdist}} can be written in terms of a sum of iid exponential random variables. In particular, it is straightforward to check that
\begin{align*}
Y_{k}
=_{\dd}-\ln(E_{1}+\dots+E_{k+1}),
\end{align*}
where $=_{\dd}$ denotes equality in distribution and $E_{1},\dots,E_{k+1}$ are iid exponential random variables with unit rate (i.e.\ $\E[E_{j}]=1$).

The next theorem yields asymptotic expansions for the moments of the slowest FPTs by proving the moment convergence,
\begin{align}\label{mdb}
\E[(\lambda T_{N-k,N}-b_{N})^{m}]
\to\E[Y_{k}^{m}]\quad\text{as }N\to\infty,
\end{align}
for any moment $m\in\{1,2,\dots\}$ and fixed $k\ge0$. In light of the convergence in distribution in Theorem~\ref{expdist}, it is natural to expect that \eqref{mdb} holds. However, \eqref{mdb} is not a corollary of Theorem~\ref{expdist}, since convergence in distribution does not imply convergence of moments. {Furthermore, we note that \cite{pickands1968} proved that the convergence in distribution in Theorem~\ref{expdist} implies that \eqref{mdb} holds for $k=0$ (i.e.\ for the slowest FPT $T_{N,N}$). However,} to our knowledge there is no previous result that shows that the convergence in distribution in Theorem~\ref{expdist} implies that \eqref{mdb} holds for any fixed $k\ge1$. Indeed, we show in section~\ref{powerlawdecay} below that moment convergence for $T_{N,N}$ is not equivalent to moment convergence of $T_{N-k,N}$. We prove Theorem~\ref{expmoments} by showing that the sequence of random variables, $\{(\lambda T_{N-k,N}-b_{N})^{2m}\}_{N\ge1}$, is uniformly integrable.

\begin{theorem}\label{expmoments}
Assume
\begin{align}\label{aless}
S(t)
\sim A(\lambda t)^{-p}e^{-\lambda t}\quad\text{as }t\to\infty,
\end{align}
for some $A>0$, $\lambda>0$, and $p\in\R$. Then for each moment $m\in\{1,2,\dots\}$ and any fixed $k\ge0$, we have that
\begin{align}\label{cmt}
\E\big[(\lambda T_{N-k,N}-b_{N})^{m}\big]
\to\E[Y_{k}^{m}]\quad\text{as }N\to\infty,
\end{align}
where $Y_{k}$ and $b_{N}$ are as in Theorem~\ref{expdist}. In particular, the mean satisfies
\begin{align*}
\E[T_{N-k,N}]
&=\lambda^{-1}\big(b_{N}+\E[Y_{k}]+o(1)\big)
=\lambda^{-1}\big(b_{N}+\gamma-H_{k}+o(1)\big)\\
&=\lambda^{-1}\big(\ln N-p\ln\ln N+\ln A+\gamma-H_{k}+o(1)\big),\quad\text{as }N\to\infty,
\end{align*}
where $b_{N}$ is given by \eqref{bN}, $\gamma\approx0.5772$ is the Euler-Mascheroni constant, and $H_{k}=\sum_{r=1}^{k}r^{-1}$ is the $k$-th harmonic number. Further, the variance satisfies
\begin{align*}
\textup{Variance}(T_{N-k,N})
&=\lambda^{-2}\big(\psi'(k+1)+o(1)\big),\quad\text{as }N\to\infty,
\end{align*}
where $\psi'(k+1)=\sum_{r=0}^{\infty}(r+k+1)^{-2}$ is the first order polygamma function.
\end{theorem}

Note that Theorem~\ref{expmoments} implies that the coefficient of variation of $T_{N-k,N}$ vanishes as $N\to\infty$,
\begin{align}\label{cvvanish}
\frac{\sqrt{\textup{Variance}(T_{N-k,N})}}{\E[T_{N-k,N}]}
\sim\frac{\sqrt{\psi'(k+1)}}{\ln N}\quad\text{as }N\to\infty.
\end{align}


Theorem~\ref{expdist} approximates the distribution of the $(k+1)$st slowest FPT. The next result, which is a corollary of Theorem~\ref{expdist}, generalizes Theorem~\ref{expdist} to approximate the joint distribution of the $k+1$ slowest FPTs. 

\begin{corollary}\label{expjoint}
Assume
\begin{align*}
S(t)
\sim A(\lambda t)^{-p}e^{-\lambda t}\quad\text{as }t\to\infty,
\end{align*}
for some $A>0$, $\lambda>0$, and $p\in\R$. For each fixed $k\ge0$, we have the following convergence in distribution for the joint random variables,
\begin{align}\label{cd2}
\Big(\lambda T_{N-k,N}-b_{N},\lambda T_{N-k+1,N}-b_{N},\dots,\lambda T_{N,N}-b_{N}\Big)
\to_{\textup{d}}
\mathbf{Y}\;\;\text{as }N\to\infty,
\end{align}
where $b_{N}$ is as in Theorem~\ref{expdist} and $\mathbf{Y}\in\R^{k+1}$ is the random vector,
\begin{align*}
 \mathbf{Y}
=\Big(-\ln(E_{1}+\dots+E_{k+1}),-\ln(E_{1}+\dots+E_{k}),\dots,-\ln (E_{1})\Big)\in\R^{k+1},
\end{align*}
where $E_{1},\dots,E_{k+1}$ are iid unit rate exponential random variables.
\end{corollary}

\subsection{Accelerating convergence}\label{bestdecay}

Theorem~\ref{expmoments} above gives the first few terms in the asymptotic expansion of the $m$th moment of $T_{N-k,N}$ for $N\gg k+1\ge1$ assuming only that $S(t)$ decays exponentially. In many examples of interest (see section~\ref{examples}), $S(t)$ is a sum of exponentials,
\begin{align*}
S(t)
=Ae^{-\lambda t}+\sum_{n}B_{n}e^{-\beta_{n}t},
\end{align*}
where $0<\lambda<\beta_{1}<\beta_{2}<\cdots$. For this case, the following result gives higher order asymptotic expansions of the $m$th moment of $T_{N-k,N}$ for $N\gg k+1\ge1$. The proof relies on Renyi's representation of exponential order statistics \cite{renyi1953} and some detailed asymptotic estimates.

\begin{theorem}\label{work}
Assume 
\begin{align}\label{assump}
S(t)-Ae^{-\lambda t}
=O(e^{-\beta t})\quad\text{as }t\to\infty,
\end{align}
where $A>0$ and $\beta>\lambda>0$. Then for any moment $m\in(0,\infty)$ and any fixed $k\ge0$, we have 
\begin{align*}
&\E[(T_{N-k,N}-\lambda^{-1}\ln A)^{m}]\\
&\quad=\frac{1}{\lambda^{m}}\bigg(\E\bigg[\Big(\sum_{i=1}^{N-k}\frac{E_{i}}{i+k}\Big)^{m}\bigg]+O\big(N^{-(\beta/\lambda-1)}(\ln N)^{m-1}\big)\bigg)\quad\text{as }N\to\infty,
\end{align*}
where $E_{1},E_{2},\dots,E_{N-k}$ are iid unit rate exponential random variables. In particular,
\begin{align*}
\E[T_{N-k,N}]
=\frac{1}{\lambda}\Big(H_{N}-H_{k}+\ln A+O(N^{-(\beta/\lambda-1)})\Big)\quad\text{as }N\to\infty,
\end{align*}
where $H_{n}=\sum_{r=1}^{n}r^{-1}$ is the $n$-th harmonic number. Further, as $N\to\infty$,
\begin{align*}
\textup{Variance}(T_{N-k,N})
&=\frac{1}{\lambda^{2}}\Big(\sum_{i=1}^{N-k}\frac{1}{(i+k)^{2}}+O(N^{-(\beta/\lambda-1)}\ln N)\Big),\\
&=\frac{1}{\lambda^{2}}\Big(\psi'(k+1)-\psi'(N+1)+O(N^{-(\beta/\lambda-1)}\ln N)\Big),
\end{align*}
where $\psi'(k+1)=\sum_{r=0}^{\infty}(r+k+1)^{-2}$ is the first order polygamma function.
\end{theorem}

Theorem~\ref{work} gives better estimates of the moments of $T_{N-k,N}$ than Theorem~\ref{expmoments} in the case that \eqref{assump} holds rather than merely that \eqref{aless} holds with $p=0$. Therefore, assuming \eqref{assump}, Theorem~\ref{work} suggests that we choose the sequence $\{b_{N}\}_{N\ge1}$ in Theorem~\ref{expdist} to be
\begin{align}
\begin{split}
\label{betterb}
b_{N}
&=H_{N}-\gamma+\ln A\\
&=\ln N+\ln A+\frac{1}{2N}-\sum_{j=1}^{\infty}\frac{B_{2j}}{2jN^{2j}},
\end{split}
\end{align}
where $H_{N}=\sum_{r=1}^{N}r^{-1}$ is the $N$-th harmonic number, $\gamma\approx0.5772$ is the Euler-Mascheroni constant, and $B_{j}$ are the Bernoulli numbers. The second equality in \eqref{betterb} follows from expanding $H_{N}$. It follows immediately from \eqref{bnp} that the convergence in distribution in Theorem~\ref{expdist} holds with the choice in \eqref{betterb} assuming \eqref{assump}. Furthermore, Theorem~\ref{work} implies that this convergence in distribution occurs with a faster rate, in the sense that our bound on the mean of the difference $\lambda T_{N-k,N}-b_{N}-Y_{k}$ vanishes at a faster rate as $N\to\infty$ with $b_{N}$ in \eqref{betterb} rather than \eqref{bNalt}.

\subsection{Power law decay}\label{powerlawdecay}

We now present analogs of the results in section~\ref{exponentialdecay} for the case that the survival probability of a single FPT $\tau$ vanishes according to a power law.

\begin{theorem}\label{powerdist}
Assume
\begin{align}\label{pl}
S(t)
\sim 
(\lambda t)^{-p}\quad\text{as }t\to\infty,
\end{align}
for some $\lambda>0$ and $p>0$. For any fixed integer $k\ge0$, we have that
\begin{align*}
\frac{\lambda T_{N-k,N}}{N^{1/p}}
\to_{\dd}Z_{k}\quad\text{as }N\to\infty,
\end{align*}
where $Z_{k}$ has probability distribution,
\begin{align*}
\P(Z_{k}\le {z})
=
\frac{\Gamma(k+1,z^{-p})}{k!}, & \quad \text{if }{z}>0,
\end{align*}
and $\P(Z_{k}\le z)=0$ if $z\le0$,  where $\Gamma(a,x):=\int_{x}^{\infty}u^{a-1}e^{-u}\,\dd u$ denotes the upper incomplete gamma function.

Furthermore, for each fixed $k\ge0$, we have the following convergence in distribution for the joint random variables,
\begin{align}\label{cd4}
\left(\frac{\lambda T_{N-k,N}}{N^{1/p}},\frac{\lambda T_{N-k+1,N}}{N^{1/p}},\dots,\frac{\lambda T_{N,N}}{N^{1/p}}\right)
\to_{\textup{d}}
\mathbf{Z}
\in\R^{k+1}\quad\text{as }N\to\infty,
\end{align}
where $\mathbf{Z}\in\R^{k+1}$ is the random vector,
\begin{align*}
 \mathbf{Z}
=\Big((E_{1}+\dots+E_{k+1})^{-1/p},(E_{1}+\dots+E_{k})^{-1/p},\dots,(E_{1})^{-1/p}\Big)\in\R^{k+1},
\end{align*}
where $E_{1},\dots,E_{k+1}$ are iid unit rate exponential random variables.
\end{theorem}

Roughly speaking, Theorem~\ref{powerdist} implies
\begin{align*}
\P(T_{N-k,N}\le t)
\approx
\P(Z_{k}\le N^{-1/p}\lambda t)
=\frac{\Gamma(k+1,N(\lambda t)^{-p})}{k!},\quad\text{if }{{N\gg k+1\ge1}}.
\end{align*}
In the case that $k=0$, the limiting distribution is Frechet with shape $p>0$, $\P(Z_{0}\le {z})=\exp(-z^{-p})$, and so
\begin{align*}
\P(T_{N,N}\le t)
\approx
\P(Z_{0}\le N^{-1/p}\lambda t)
=\exp(-N(\lambda t)^{-p}),\quad\text{if }N\gg 1.
\end{align*}
We further note that, as is evident from the statement of Theorem~\ref{powerdist}, 
\begin{align*}
Z_{k}
=_{\dd}(E_{1}+\dots+E_{k+1})^{-1/p},\quad k\ge0,
\end{align*}
where $E_{1},\dots,E_{k+1}$ are iid unit rate exponential random variables.

The next theorem approximates the moments of $T_{N-k,N}$ assuming $S(t)$ decays according to the power law in \eqref{pl}. For such a power law decay, it follows immediately from \eqref{unwieldy} that 
\begin{align*}
\E[(T_{N-k,N})^{m}]
=\infty\quad\text{if }m\ge(k+1)p.
\end{align*}
Hence, the next theorem assumes $0<m<(k+1)p$. As described in section~\ref{exponentialdecay}, moment convergence does not in general follow from convergence in distribution. We prove Theorem~\ref{powermoments} by showing that the sequence of random variables, $\{(\lambda T_{N-k,N}N^{-1/p})^{r}\}_{N\ge1}$ is uniformly integrable for any even integer $r\ge2$. The proof is quite different from the proof of Theorem~\ref{expmoments}, where the difference stems from the fact that $\E[(T_{N,N})^{m}]=\infty$ and $\E[(T_{N-k,N})^{m}]<\infty$ if $p\le m<(k+1)p$.

\begin{theorem}\label{powermoments}
Assume
\begin{align*}
S(t)
\sim 
(\lambda t)^{-p}\quad\text{as }t\to\infty,
\end{align*}
for some $\lambda>0$ and $p>0$. Then for any fixed $k\ge0$ and moment $m$ such that
\begin{align*}
0<m<(k+1)p,
\end{align*}
we have that
\begin{align*}
\E[(T_{N-k,N})^{m}]
&\sim \lambda^{-m}N^{m/p}\frac{\Gamma(k+1-m/p)}{\Gamma(k+1)},\quad\text{as }N\to\infty.
\end{align*}
\end{theorem}

A counterintuitive implication of Theorem~\ref{powermoments} is that
\begin{align}\label{cii}
E[T_{N-k,N}]
<\E[\tau]=\infty\quad\text{if }p\le1<(k+1)p\le Np.
\end{align}
This means that, for example, if $p=1$, then the FPT of any given searcher has infinite mean (i.e.\ $\E[\tau]=\infty$), but the FPT of the second slowest searcher out of $N$ searchers has a finite mean (i.e.\ $\E[T_{N-1,N}]<\infty$) for any $N\ge2$. The result \eqref{cii} is especially counterintuitive if $N\gg1$. To check \eqref{cii} in a simple special case, observe that \eqref{unwieldy} implies that for $k=p=1$, 
\begin{align*}
\E[T_{N-1,N}]
=\int_{0}^{\infty}\P(T_{N-1,N}>t)\,\dd t
&=\sum_{i=0}^{N-2}{N\choose i}\int_{0}^{\infty}(S(t))^{N-i}(1-S(t))^{i}\,\dd t,
\end{align*}
and the slowest decaying integrand in the sum decays like $t^{-2}$ as $t\to\infty$, and thus all the terms in the sum are finite. We discuss \eqref{cii} in section~\ref{halfline} in the context of diffusing searchers in an unbounded spatial domain.

\section{Menopause timing}\label{meno}

We now apply the general mathematical results of section~\ref{math} to {a} model of ovarian aging and menopause timing. As described in the Introduction, a woman is born with around $N\approx5\times10^{5}$ PFs in her ovarian reserve \cite{wallace2010}. The number of PFs in her reserve then {decays} over the next 40-60 years of her life. Menopause occurs when her reserve drops to around $k\approx10^{3}$ PFs, which is around age 50 years for most women. Hence, {letting} $\{\tau_{n}\}_{n=1}^{N}$ denote the times that each of the $N$ PFs leave the reserve, menopause occurs at time
\begin{align}\label{menopausetime}
T_{N-k,N}\quad\text{for }N\gg k=10^{3}.
\end{align}

{
We can apply the results of section~\ref{exponentialdecay} to any model of ovarian aging if the model assumes
\begin{align}\label{iid}
\{\tau_{n}\}_{n=1}^{N}\text{ are iid},
\end{align}
and
\begin{align}\label{mexp}
S(t):=\P(\tau_{n}>t)\sim Ae^{-\lambda t}\quad\text{as }t\to\infty.
\end{align}
Models of PF dynamics which assume \eqref{iid}-\eqref{mexp} have a long history. Perhaps the earliest models that assume \eqref{iid}-\eqref{mexp} are those of Faddy, Jones, and Edwards \cite{faddy1976} and Faddy, Gosden, and Edwards \cite{faddy1983}, which were models of PF dynamics in mice. Faddy and Gosden \cite{faddy1995} proposed and analyzed a stochastic model of PF dynamics in women which assumed \eqref{iid}-\eqref{mexp}, where the parameters $A>0$ and $\lambda>0$ in \eqref{mexp} were chosen by fitting to PF data. PF decay in a woman's ovarian reserve as been compared to radioactive decay, which satisfies \eqref{iid}-\eqref{mexp}. Specifically, the review by Hirshfield \cite{hirshfield1991} discussed the possibility that PF growth activation is a ``randomized stochastic event, similar to radioactive decay,'' and the book by Finch and Kirkwood \cite{finch2000} discussed ``pure chance'' PF growth activation and discussed the similarity of the decay of PFs in the reserve to radioactive decay. {More recently, experimental results obtained by one of our groups revealed that the integrated stress response (ISR) pathway influences the probability that an individual PF {leaves the reserve} \cite{llerena2021}. Since ISR activity fluctuates over time in a single PF and varies broadly between PFs \cite{hagen2022}, the ISR activity in individual PFs {was recently} modeled by independent random walks, where a single PF {leaves the reserve} when its ISR activity crosses a given threshold \cite{johnson2022}.  
}
}

{Following the ovarian aging model in \cite{johnson2022},} if $X(t)$ denotes the ISR activity of a given PF at time $t\ge0$, then $X$ evolves according to the stochastic differential equation (SDE),
\begin{align}\label{sder}
\dd X
=-V\,\dd t +\sqrt{2D}\,\dd W,\quad X(0)=x,
\end{align}
where $W=\{W(t)\}_{t\ge0}$ is a standard Brownian motion. In \eqref{sder}, $V>0$ is a drift parameter which describes the tendency of the ISR activity $X$ to decrease over time, and $D>0$ is a diffusivity parameter which describes the size of the stochastic fluctuations. {The PF leaves the reserve when $X$ crosses at threshold either at $X=0$ or $X=L$}. Hence, the ovarian reserve exit times $\{\tau_{n}\}_{n=1}^{N}$ in this model are iid realizations of the FPT,
\begin{align}\label{reserveexit}
\tau
:=\inf\{t>0:X(t)\notin(0,L)\}.
\end{align}

By solving the associated backward Kolmogorov equation, the survival probability of a single FPT $\tau$ in \eqref{reserveexit} can be shown to have the following form (see equation~(19) in the Appendix in \cite{johnson2022}),
\begin{align}\label{sre}
S(t)
:=\P(\tau>t)
=Ae^{-\lambda t}
+\sum_{k\ge2}A_{k}e^{-\lambda_{k}t}.
\end{align}
In \eqref{sre}, $A>0$, $A_{k}\in\R$ for $k\ge2$, and 
\begin{align}\label{order0}
0<\lambda<\lambda_{2}<\lambda_{3}<\cdots
\end{align}
depend on the parameters in \eqref{sder}-\eqref{reserveexit}. In particular, the constants in the leading order term in \eqref{sre} are
\begin{align}\label{lambdaA}
\begin{split}
\lambda
&=\frac{D\pi^{2}}{L^{2}}+\frac{V^{2}}{4D},\quad
A
=\frac{4 \sqrt{2} \pi  D^2 \sqrt{L} \big(1- e^{\frac{-L V}{2 D}}\big)}{4 \pi ^2 D^2+L^2 V^2}e^{\frac{V}{2D}x}\sqrt{\frac{2}{L}}\sin\Big(\frac{\pi x}{L}\Big).
\end{split}
\end{align}

Owing to the ordering in \eqref{order0}, the survival probability in \eqref{sre} decays exponentially at large time,
\begin{align*}
S(t)-Ae^{-\lambda t}
=O(e^{-\lambda_{2}t})\quad\text{as }t\to\infty,
\end{align*}
where $0<\lambda<\lambda_{2}$ and $A>0$ are in \eqref{order0}-\eqref{lambdaA}. Therefore, we can apply Theorems~\ref{expdist}, \ref{expmoments}, and \ref{work} to approximate the probability distribution and moments of the time of menopause for a given woman, $T_{N-k,N}$ in \eqref{menopausetime}. {
We note that one can immediately generalize this calculation to the case of stochastic initial conditions. In particular, if the initial distribution of each searcher has measure $\mu$, then one merely replaces $A$ in \eqref{lambdaA} by $\int_{0}^{L}A\,\dd\mu(x)$.
}

\begin{figure}
  \centering
             \includegraphics[width=1\textwidth]{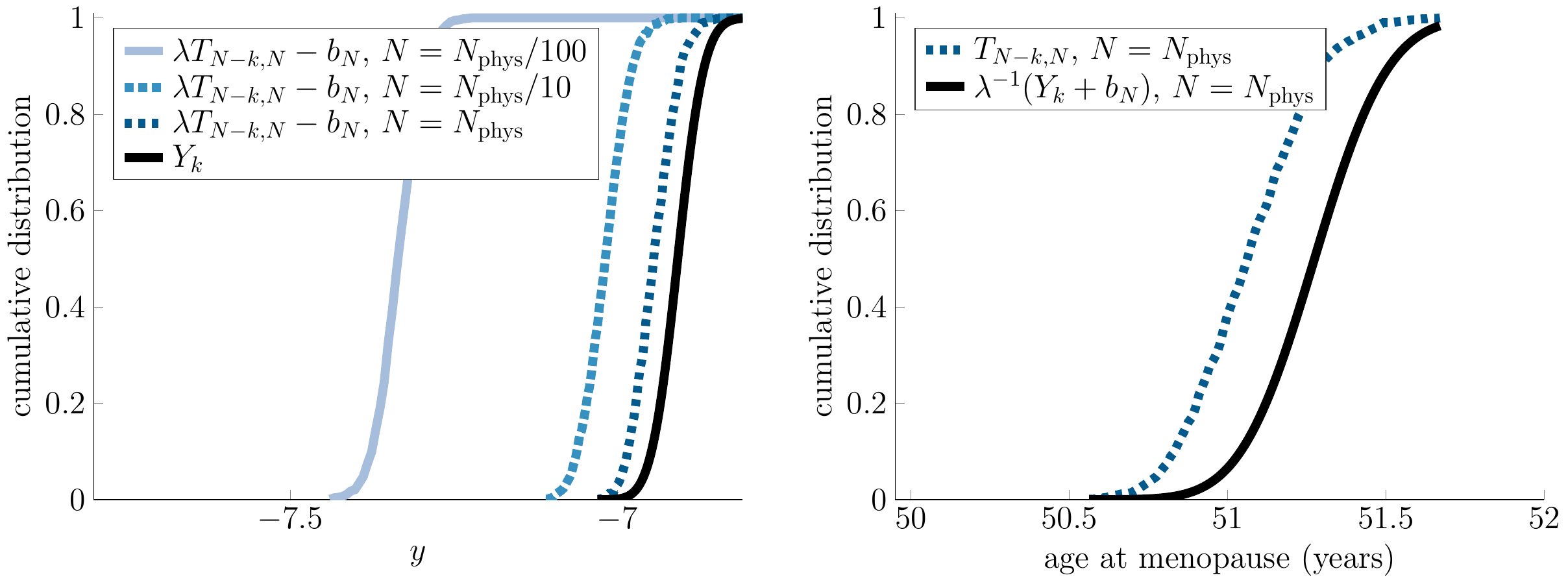}
 \caption{Convergence in distribution implied by Theorem~\ref{expdist} for a model of ovarian aging. See the text for details.}
 \label{figmeno}
\end{figure}

In Figure~\ref{figmeno}, we show the convergence in distribution implied by Theorem~\ref{expdist}. In the left panel of Figure~\ref{figmeno}, we plot the distribution of the rescaled and shifted slowest FPT, 
\begin{align}\label{cdfsim}
\P(\lambda T_{N-k,N}-b_{N}\le y),
\end{align}
for $k=10^{3}$ and $b_{N}=\ln(AN)$, as the number of PFs increases up to a physiological value of \cite{wallace2010}
\begin{align}\label{Nphys}
N
=N_{\textup{phys}}
:=3.2\times10^{5}.
\end{align} 
The solid black curve in the left panel of Figure~\ref{figmeno} is
\begin{align}\label{cdftheory}
\P(Y_{k}\le y)
=\frac{\Gamma(k+1,e^{-y})}{k!},
\end{align}
as in Theorem~\ref{expdist}. The convergence of \eqref{cdfsim} to \eqref{cdftheory} as $N$ increases is evident in this plot. In this plot, we take the following parameter values,
\begin{align}\label{param}
D/x^{2}
=0.004\,\textup{year}^{-1},\quad
V/x
=0.051\,\textup{year}^{-1},\quad
x/L
=1/2,
\end{align}
which were obtained in \cite{johnson2022} by fitting \eqref{sre} to histological data of PF decay \cite{wallace2010}. See section~\ref{appmeno} in the Appendix for details on the numerical method used to obtain the cumulative distribution function in \eqref{cdfsim}. 

In the right panel of Figure~\ref{figmeno}, we compare the distribution of the time to menopause, $T_{N-k,N}$, for the physiologically relevant PF number, $N=N_{\textup{phys}}$. This plot also shows the distribution of the theoretical approximation,
\begin{align*}
\lambda^{-1}(Y_{k}+b_{N}),
\end{align*}
where $Y_{k}$ is in \eqref{cdftheory} (as in Theorem~\ref{expdist}), $b_{N}=\ln(AN)$, and $\lambda$ and $A$ are in \eqref{lambdaA}. This figure shows that the theoretical approximation to the time of menopause is accurate to about 2 or 3 months. Indeed, the mean of $T_{N-k,N}$ computed from stochastic simulations is
\begin{align*}
\E[T_{N-k,N}]
=51.07\,\textup{years},
\end{align*}
which is within 3 months of the mean of $\E[T_{N-k,N}]$ estimated from ignoring the $o(1)$ terms in Theorem~\ref{expmoments},
\begin{align}\label{slowlog}
\E[T_{N-k,N}]
\approx
\lambda^{-1}(\ln N+\ln A+\gamma-H_{k})
=51.27\,\textup{years}.
\end{align}

{
We now comment on how this analysis offers potential explanations for some perplexing aspects of ovarian biology. First,} these results demonstrate the utility of the very large and seemingly redundant number of PFs. As described in the Introduction, the number $N$ of PFs is a few orders of magnitude greater than the number that will be ovulated or the number engaged in ovarian endocrine function and menstrual signaling. What explains this {so-called ``wasteful'' oversupply \cite{faddy1995, themmen2005, hartshorne2009, inserra2014, albamonte2019}?} These results show that the very large number $N$ of PFs ensures that there will be a supply of PFs available for ovulation for several decades of a woman's life. Indeed, for the parameters in \eqref{param}, one can compute that a typical PF spends less than 20 years in the reserve, {$\E[\tau]\approx19\,\textup{years}$.} Despite the relatively short time {that} a typical PF spends in the reserve, the large number $N$ of PFs ensures that the ovarian reserve lasts around 50 years. In fact, the slow logarithmic growth of $\E[T_{N-k,N}]$ as a function of $N$ means that the number of PFs $N$ must be on the order of hundreds of thousands to extend the lifetime of the reserve much beyond the time {$\E[\tau]\approx19\,\textup{years}$.}

{
Second, despite the enormous variability in the PF starting supply $N$ across a population, the menopause age distribution is quite narrow across a population. Indeed, in a dataset of only 14 girls at birth, the largest $N$ was over 500\% greater than the smallest $N$ (namely, $N=10^{6}$ versus $N=1.5\times10^{5}$) \cite{wallace2010}. In contrast, the menopause age varies by at most around $50\%$ between healthy women (age 40 to 60 years) \cite{weinstein2003}. This discrepancy can be understood immediately from \eqref{slowlog}, since this equation predicts that menopause age depends logarithmically on $N$. Indeed, taking $N=1.5\times10^{5}$ versus $N=10^{6}$ in \eqref{slowlog} yields respective menopause ages of 47 and 58 years.

Third, a unilateral oophorectomy (removal of a single ovary) tends to yield only a slightly earlier menopause age. Numerical estimates vary \cite{yasui2012, bjelland2014, rosendahl2017}, but a unilateral oophorectomy is associated with an earlier menopause age of at most a couple of years. As noted by \cite{bjelland2014}, since removing an ovary cuts the PF count in half, it is counterintuitive that the menopause age ``penalty'' for a unilateral oophorectomy is so small. Furthermore, this penalty is at most only very weakly correlated with the woman's age at the time of unilateral oophorectomy (see Figure~3 in \cite{rosendahl2017}). Both of these observations are consistent with the analysis above. Indeed, \eqref{slowlog} predicts that removing half of the PFs at birth causes a menopause age penalty of only $\ln(2)/\lambda\approx4.0\,\text{years}$, and the iid assumption in \eqref{iid} implies that this penalty is independent of when the ovary is removed (assuming merely that there are at least $k=10^{3}$ {PFs} in the remaining ovary at the time of oophorectomy). {We note that including population heterogeneity in the parameters in \eqref{sder} as in \cite{johnson2022} (see equation (10) therein) would imply that the menopause age penalty estimate of $\ln(2)/\lambda\approx4.0\,\text{years}$ would vary by a fraction of a year between women.}

In addition, the analysis above predicts an interesting consequence of the menopause threshold $k\approx10^{3}\gg1$. In particular, although the model assumes each PF leaves the reserve at an independent random time, the resulting menopause age is nearly deterministic (for a given woman with a fixed $N$). Indeed, Theorem~\ref{expmoments} implies that the menopause age standard deviation is
\begin{align}\label{stddevmeno}
\sqrt{\textup{Variance}(T_{N,N-k})}
\approx\lambda^{-1}\sqrt{\psi'(k+1)},
\end{align}
and taking $k=10^{3}$ in \eqref{stddevmeno} yields a standard deviation of only 2.2 months. In contrast, taking $k=0$ in \eqref{stddevmeno} (i.e.\ menopause occurs when the PF supply is completely exhausted) increases the standard deviation to over 7 years.

Naturally, the analysis above makes a number of simplifying assumptions. For instance, though the iid assumption in \eqref{iid} is common in models of ovarian aging \cite{faddy1976, faddy1983, faddy1995}, neighboring PFs in the ovary may be correlated due to physiological processes that fluctuate over time regionally within the ovary \cite{llerena2021}. {Further, mechanistic knowledge has been accumulating on the ovarian reserve establishment and PF activation \cite{grosbois2020}, and these details are not directly accounted for in the simple iid assumption in \eqref{iid}.} In addition, though it is common to assume a link between the initial PF supply and the menopause age \cite{faddy1996, wallace2010}, some data in mice has cast doubt on this link \cite{bristol2006}. It is likely also that the rate of loss of PFs can be modified by known or unidentified exposures in individuals, but modeling those individual cases will need to be informed by those specific circumstances of exposure type and time.
} {An additional limitation is that the analysis above does not distinguish between (i) PFs which exit the reserve due to growth activation and (ii) PFs which exit the reserve due to atresia. This is in contrast to some prior models which track follicles through multiple stages of development with distinct growth and death rates which are piecewise constant in time  \cite{faddy1976, faddy1983, faddy1995}. Determining the relative contribution of activation versus atresia is critical to jointly predict the numbers of primordial and growing follicles and thus remains an important area for research.
}

\section{Stochastic search}\label{examples}

In this section, we use the general mathematical results of section~\ref{math} to investigate several prototypical models of stochastic search and compare the asymptotic theory to numerical simulations. The details of the numerical simulation methods are given in section~\ref{appnumerical} in the Appendix. 

\subsection{Diffusive escape from an interval}\label{interval}

Let $\{X(t)\}_{t\ge0}$ be a one-dimensional, pure diffusion process with diffusivity $D>0$. Let $\tau$ be the FPT for the diffusion to escape the interval $(-L,L)$,
\begin{align*}
\tau
:=\inf\{t>0:X(t)\notin(-L,L)\}.
\end{align*}
Assume that the searcher starts at $X(0)=x_{0}\in[0,L)$. 
Let $\{\tau_{n}\}_{n=1}^{N}$ be an iid sequence of $N$ realizations of $\tau$, representing the FPTs of $N$ iid searchers. 

A standard eigenfunction analysis of the associated backward Kolmogorov equation yields that the survival probability, $S(t):=\P(\tau>t)$, decays exponentially at large time (see section~\ref{appinterval} in the Appendix),
\begin{align*}
S(t)-Ae^{-\lambda t}
=O(e^{-\beta t})\quad\text{as }t\to\infty,
\end{align*}
where $\beta=4\lambda$ and 
\begin{align}\label{lA}
\lambda
=\frac{\pi^{2}}{4}\frac{D}{L^{2}},\quad
A=\frac{4}{\pi}\sin(\pi(x_{0}+L)/(2L)).
\end{align}
Theorems~\ref{expdist}, \ref{expmoments}, and \ref{work} thus yield the distribution and moments of the $(k+1)$st slowest FPT, $T_{N-k,N}$ if ${{N\gg k+1\ge1}}$. In particular, Theorem~\ref{work} implies that the mean FPT of the slowest searcher satisfies
\begin{align}\label{exm}
\E[T_{N,N}]
&=\frac{1}{\lambda}\Big(H_{N}+\ln A+O(N^{-3})\Big),\quad\text{as }N\to\infty.
\end{align}
{We note that it is immediate to generalize this calculation to the case of stochastic initial conditions. In particular, if the initial distribution of each searcher has measure $\mu$, then one merely replaces $A$ in \eqref{lA} by (see section~\ref{appinterval} in the Appendix for more details)
\begin{align*}
\int_{-L}^{L}\frac{4}{\pi}\sin(\pi(x_{0}+L)/(2L))\,\dd \mu(x_{0}).
\end{align*}}

\begin{figure}
  \centering
             \includegraphics[width=1\textwidth]{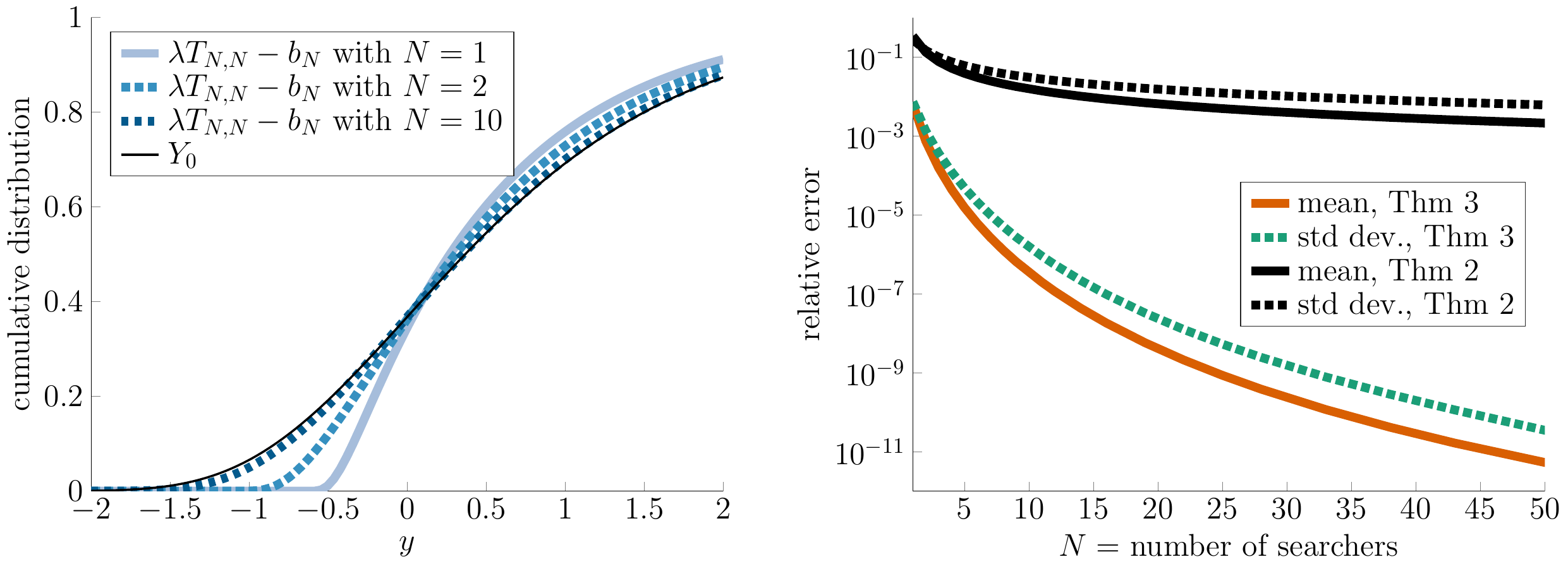}
 \caption{Convergence in distribution and moments for diffusive escape from an interval. The left panel plots the distribution of the rescaled slowest FPT, $\lambda T_{N,N}-b_{N}$, with $b_{N}$ in \eqref{betterb} for $N=1,2,10$, which converges rapidly to the distribution of $Y_{0}$ defined in Theorem~\ref{expdist}. The right panel plots the relative errors in \eqref{remean}-\eqref{restd} for computing the mean and standard deviation of $T_{N,N}$ via Theorem~\ref{work}.}
 \label{figinterval}
\end{figure}

In Figure~\ref{figinterval}, we illustrate the conclusions of Theorems~\ref{expdist}, \ref{expmoments}, and \ref{work} for this example. In the left panel of Figure~\ref{figinterval}, we plot the distribution of $\lambda T_{N,N}-b_{N}$ where $b_{N}=H_{N}-\gamma+\ln A$ (see \eqref{betterb}) for $N=1,2,10$. This plot shows that $\lambda T_{N,N}-b_{N}$ converges in distribution very rapidly to the random variable $Y_{0}$ defined in Theorem~\ref{expdist} as $N$ increases. In the right panel, we plot the relative error {for} the approximations in Theorem~\ref{expmoments} and Theorem~\ref{work} {for} the mean and standard deviation of $T_{N,N}$. That is, we plot the following relative errors for Theorem~\ref{expmoments} (black curves),
\begin{align}\label{remean}
\bigg|\frac{\E[T_{N,N}]-\lambda^{-1}(\ln(AN)+\gamma)}{\E[T_{N,N}]}\bigg|,\;
\bigg|\frac{\sqrt{\textup{Var}(T_{N,N})}-\lambda^{-1}(\pi^{2}/6)}{\sqrt{\textup{Var}(T_{N,N})}}\bigg|
\end{align}
and the following relative errors for Theorem~\ref{work} (orange and green curves),
\begin{align}\label{restd}
\bigg|\frac{\E[T_{N,N}]-\lambda^{-1}(H_{N}+\ln A)}{\E[T_{N,N}]}\bigg|,\;
\bigg|\frac{\sqrt{\textup{Var}(T_{N,N})}-\lambda^{-1}(\pi^{2}/6-\psi'(N+1))}{\sqrt{\textup{Var}(T_{N,N})}}\bigg|.
\end{align}
As expected, the error for the approximations in Theorem~\ref{work} vanish very quickly. 
In Figure~\ref{figinterval}, the distribution and statistics of $T_{N,N}$ are computed from the survival probability of $\tau$, which can be obtained via a standard eigenfunction calculation (see section~\ref{appinterval} in the Appendix for details). In Figure~\ref{figinterval}, we take $L=1/2$, $x_{0}=0$, and $D=1$. 

The rapid convergence of the approximations in Theorem~\ref{work} to the slowest FPT $T_{N,N}$ contrasts with the very slow convergence of approximations to the fastest FPT $T_{1,N}$. Indeed, the relative error for approximating the mean slowest FPT $\E[T_{N,N}]$ using Theorem~\ref{work} is less than 1\% for $N=1$, whereas a comparable relative error for approximating the mean fastest FPT $\E[T_{1,N}]$ for this one-dimensional diffusion problem is around $N=10^{6}$ \cite{lawley2020dist}.

Continuing the comparison of slowest and fastest FPTs, if we take only the leading order term in \eqref{exm} (using the expansion in \eqref{betterb}), then the mean slowest FPT has the following simple form,
\begin{align}\label{order1}
\E[T_{N,N}]
\sim\frac{1}{\lambda}\ln N
=\frac{4}{\pi^{2}}\frac{L^{2}}{D}\ln N,\quad\text{as }N\to\infty.
\end{align}
Considering only a given single searcher, we have that
\begin{align}\label{order2}
\E[\tau_{n}]
=\E[\tau]
=\frac{1}{2}\frac{L^{2}}{D}(1-x_{0}^{2}/L^{2}).
\end{align}
The fastest searcher satisfies \cite{lawley2020uni}
\begin{align}\label{order3}
\E[T_{1,N}]
\sim\frac{1}{4}\frac{L^{2}}{D}(1-x_{0}/L)^{2}\frac{1}{\ln N}\quad\text{as }N\to\infty.
\end{align}
There are a few things to notice about the order statistics in \eqref{order1}-\eqref{order3}. First, the leading order mean slowest FPT in \eqref{order1} is independent of the starting location $x_{0}\in[0,L)$. In contrast, the prefactors in the mean FPT of a single searcher in \eqref{order2} and the mean fastest FPT in \eqref{order3} depend on $x_{0}$.

Furthermore, as long as $x_{0}/L\not\approx1$, there is relatively little difference between \eqref{order1}, \eqref{order2}, and \eqref{order3}, due to the slow growth of $\ln N$ as $N$ increases. That is, for values of $N$ which are of interest in typical applications, mean FPTs for the slowest searcher in \eqref{order1}, the single searcher in \eqref{order2}, and the fastest searcher in \eqref{order3} are quite similar. Concretely, even if $N=10^{5}\gg1$, $\E[T_{N,N}]$ is only about an order of magnitude slower than $\E[\tau]$, and $\E[\tau]$ is only about an order of magnitude slower than $\E[T_{1,N}]$. Finally, we also point out that there is a form of symmetry in how the slowest FPT and fastest FPT relate to a single FPT, in that \eqref{order1}-\eqref{order3} imply
\begin{align}\label{symmetry}
\frac{\E[\tau]}{\E[T_{N,N}]}
\approx\frac{1}{\ln N}
\approx\frac{\E[T_{1,N}]}{\E[\tau]}\quad\text{if }N\gg1,
\end{align}
where the approximate equalities in \eqref{symmetry} merely ignore order one constants. 
In section~\ref{rare} below, we see that \eqref{symmetry} can break down, and we can instead have that $\E[\tau]/\E[T_{N,N}]\gg\E[T_{1,N}]/\E[\tau]$ if $N\gg1$.

\subsection{Rare {diffusive} escape (small target or deep well)}\label{rare}

In a variety of scenarios of biophysical interest, the FPT for a diffusive searcher to find a target in a bounded spatial domain satisfies
\begin{align}\label{ne}
S(t)
\sim Ae^{-\lambda t}\quad\text{as }t\to\infty,
\end{align}
where
\begin{align}\label{ne2}
A
\approx1\quad\text{and}\quad
0<\lambda\ll D/L^{2},
\end{align}
where $D>0$ is the diffusivity of the searcher and $L>0$ is the shortest distance a searcher must travel to reach the target (which assumes that searchers cannot start arbitrarily close to the target). Equations~\eqref{ne}-\eqref{ne2} mean that the time it takes most searchers to find the target is much longer than the diffusion timescale $L^{2}/D$. For example, it is well known that \eqref{ne}-\eqref{ne2} hold (i) for small targets (which is the so-called narrow escape or narrow capture problem \cite{lindsay2017}), (ii) for partially reactive targets which are small and/or have low reactivity \cite{lawley2019imp}, and (iii) if the searcher must escape a deep potential well before reaching the target \cite{williams1982}. 

The point of this section is to show what our results imply about slowest FPTs in the case that \eqref{ne}-\eqref{ne2} hold. Equations~\eqref{ne}-\eqref{ne2} imply that a single FPT $\tau$ is approximately exponentially distributed with rate $\lambda>0$, and thus
\begin{align}\label{single}
\E[\tau]\approx \lambda^{-1}\gg L^{2}/D.
\end{align}
Furthermore, it was shown in \cite{madrid2020comp} that the fastest FPT satisfies
\begin{align}\begin{split}\label{fast}
\E[T_{1,N}]
\approx
\begin{cases}
(\lambda N)^{-1}
\approx\E[\tau]/N & \text{if }N\ln N\ll(\lambda L^{2}/D)^{-1},\\
(L^{2}/D)(4\ln N)^{-1} & \text{if }N/\ln N\gg(\lambda L^{2}/D)^{-1}.
\end{cases}
\end{split}
\end{align}
In words, \eqref{single} states that the mean FPT of a single searcher, $\E[\tau]$, is much slower than the diffusion timescale of $L^{2}/D$. Further, \eqref{fast} states that the mean FPT of the fastest searcher, $\E[T_{1,N}]$, decays like $\E[\tau]/N$ for small to moderately large $N$, and $\E[T_{1,N}]$ finally decays like $(L^{2}/D)(4\ln N)^{-1}$ for very large $N$ (the behavior in \eqref{fast} is shown in Figure~\ref{figrare}).

Now, Theorem~\ref{expmoments} implies that the mean FPT of the slowest searcher, $\E[T_{N,N}]$, satisfies $\E[T_{N,N}]\sim\lambda^{-1}\ln N$ as $N\to\infty$, and thus \eqref{single}-\eqref{fast} imply
\begin{align}\label{asymm}
\frac{\E[\tau]}{\E[T_{N,N}]}
\approx\frac{1}{\ln N}
\gg\frac{\E[T_{1,N}]}{\E[\tau]}\quad\text{if }N\gg1.
\end{align}
To summarize, in the case that a typical single searcher finds the target much more slowly than the diffusion timescale, we have that (a) the fastest searcher out of $N\gg1$ searchers finds the target much faster than a single searcher, and in contrast, (b) the slowest searcher out of $N\gg1$ searchers is by comparison only slightly slower than a single searcher.

\begin{figure}
  \centering
             \includegraphics[width=.6\textwidth]{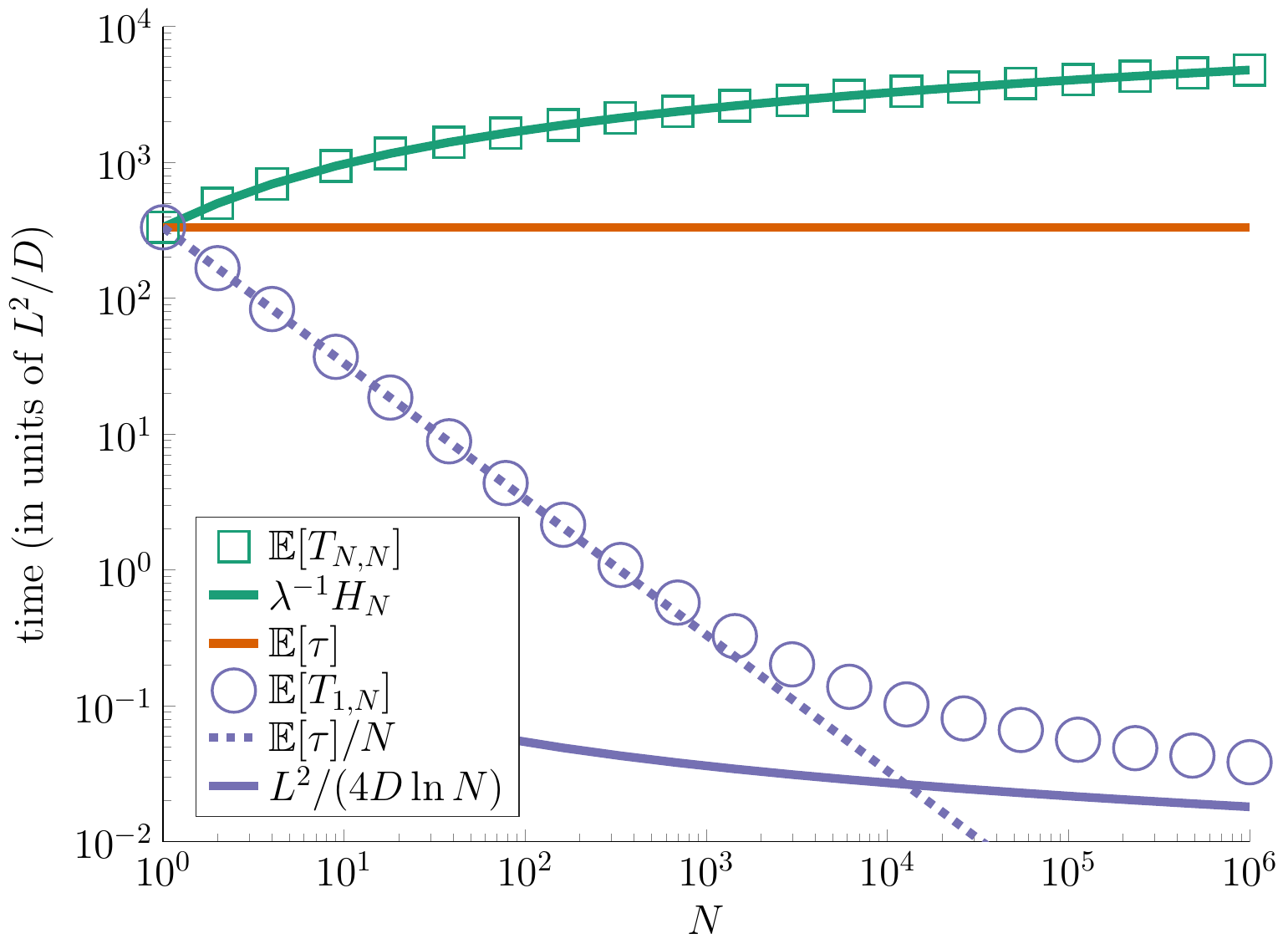}
 \caption{Slowest, fastest, and typical FPTs for a narrow capture problem. A typical single FPT ($\E[\tau]$, orange solid curve) is much slower than the diffusion timescale $L^{2}/D$, and the slowest FPT ($\E[T_{N,N}]$, green square markers) is only mildly slower {than $\E[\tau]$}. In contrast, the fastest FPT ($\E[T_{1,N}]$, purple circle markers) is much {faster} than $\E[\tau]$ and the diffusion timescale $L^{2}/D$. The green solid curve shows the leading order slowest FPT estimate from Theorem~\ref{expmoments}. The dotted and solid purple curves show the estimates in \eqref{fast}. See the text for details.}
 \label{figrare}
\end{figure}

We illustrate this point in Figure~\ref{figrare}. In this figure, we plot mean slowest and mean fastest FPTs for searchers with a small target. More specifically, we consider searchers which move by pure diffusion with diffusivity $D>0$ in a three-dimensional spherical domain of radius $L>0$. Searchers start at the boundary of the sphere, which is reflecting, and diffuse until they hit the target, which is a small sphere in the center of the domain with radius $a=L/10^{3}$. It is evident from Figure~\ref{figrare} that \eqref{asymm} holds in that the fastest FPT is much faster than a single FPT and the slowest FPT is comparatively only slightly slower than a single FPT. Indeed, for Figure~\ref{figrare} we have that
\begin{align*}
\frac{\E[\tau]}{\E[T_{N,N}]}
=7\times10^{-2}
\approx\frac{1}{\ln N}
\gg\frac{\E[T_{1,N}]}{\E[\tau]}
=1.2\times10^{-4}\quad\text{if }N=10^{6}.
\end{align*}

Finally, reiterating a point made in section~\ref{interval}, Figure~\ref{figrare} illustrates that the approximation to $\E[T_{N,N}]$ from Theorem~\ref{work} is much more accurate than the extreme value approximation to $\E[T_{1,N}]$ for finite $N$. In particular, in Figure~\ref{figrare} notice that the values of $\E[T_{N,N}]$ (green square markers) are nearly indistinguishable from the approximation $\lambda^{-1}H_{N}$ (solid green curve) from Theorem~\ref{work} (using that $\ln A\approx0$ by \eqref{ne2}) for all $N\ge1$, whereas the values of $\E[T_{1,N}]$ (purple circle markers) are relatively far from the leading order extreme value theory approximation $L^{2}/(4D\ln N)$ (solid purple curve).

To summarize, the following picture emerges if the typical search time is much slower than the diffusion timescale (i.e.\ if \eqref{ne}-\eqref{ne2} hold). Slowest searchers are only slightly slower than typical searchers, whereas fastest searchers are much faster than typical searchers. Further, the searcher starting position does not strongly impact slowest and typical FPTs, whereas fastest FPTs depend critically on searcher starting positions.

\subsection{Partially absorbing target (example due to \cite{grebenkov2022})}\label{sg}

A very interesting recent work in the chemical physics literature considered slowest FPTs of diffusing searchers in a bounded domain \cite{grebenkov2022}. By analyzing an eigenfunction expansion of the associated survival probability of a single searcher of the form,
\begin{align*}
S(t)
=Ae^{-\lambda t}+O(e^{-\beta t}),\quad\text{as }t\to\infty,
\end{align*}
where $0<\lambda<\beta$, these authors derived the following approximation for the mean slowest FPT,
\begin{align}\label{gk}
\E[T_{N,N}]
\approx\frac{1}{\lambda}\Big(\ln N+\ln A-\ln\ln\big[\tfrac{1}{1-\xi}\big]\Big),\quad\text{if }N\gg1,
\end{align}
where $\xi$ was an unknown constant near $\xi\approx0.5$. These authors found that $\xi\approx0.428$ yielded the best fit of \eqref{gk} to numerical simulations.

By Theorem~\ref{expmoments} above, we see that $-\ln\ln[\tfrac{1}{1-\xi}]$ in \eqref{gk} should be replaced by the Euler-Mascheroni constant $\gamma\approx0.5772$. Indeed, setting $\xi=0.428$ yields $-\ln\ln[\tfrac{1}{1-\xi}]=0.5823\approx\gamma$. Furthermore, Theorem~\ref{work} above yields corrections to the expansion in \eqref{gk} up to order $N^{-(\lambda/\beta-1)}$.

In Figure~\ref{figkappa}, we consider the specific example studied in \cite{grebenkov2022}, in which each searcher diffuses in a three-dimensional sphere centered on an interior spherical target that is partially absorbing (as in Figure~4 in \cite{grebenkov2022}, the domain has radius 10, the searchers start at radius 5, and the diffusion coefficient, target radius, and target reactivity are all unity). This figure illustrates the very high accuracy of the approximation $\E[T_{N,N}]\approx\lambda^{-1}(H_{N}+\ln A)$ given by Theorem~\ref{work} (solid purple curve), in which the relative error is less than $10^{-3}$ for $N=1$ and dips to nearly $10^{-14}$ by $N=8$. The details for this example are given in section~\ref{appkappa} in the Appendix.

\begin{figure}
  \centering
	\includegraphics[width=.49\textwidth]{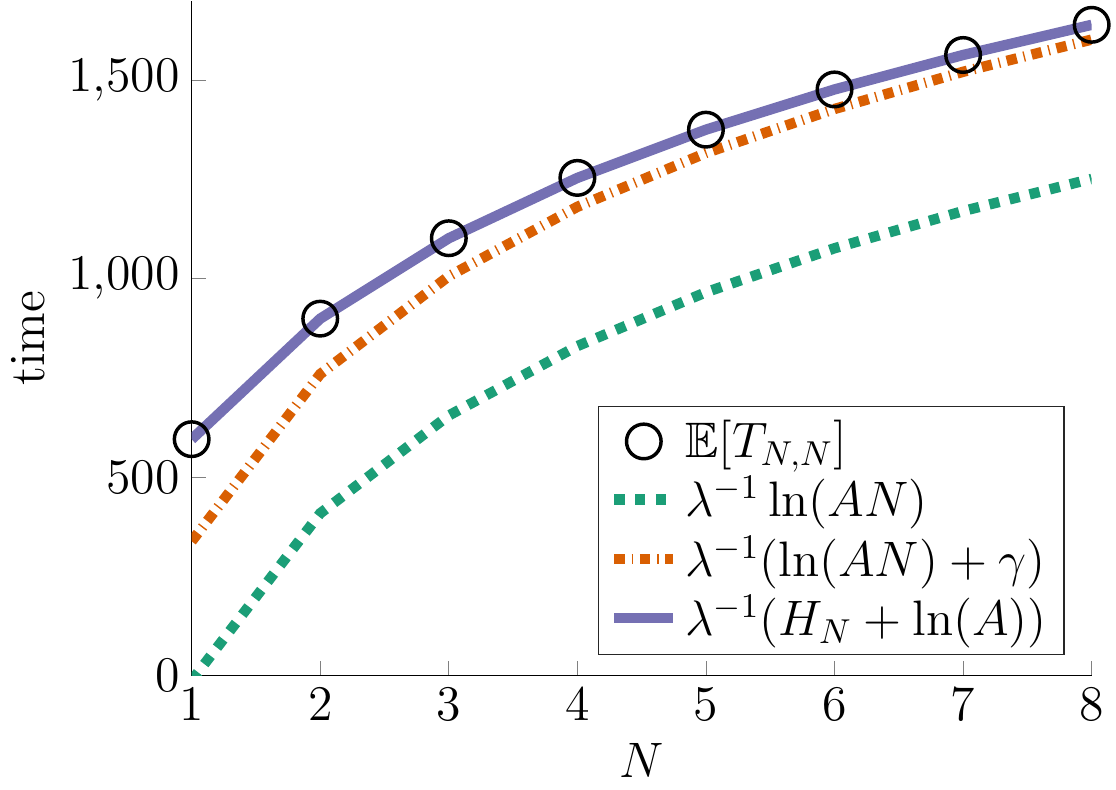}
	\includegraphics[width=.49\textwidth]{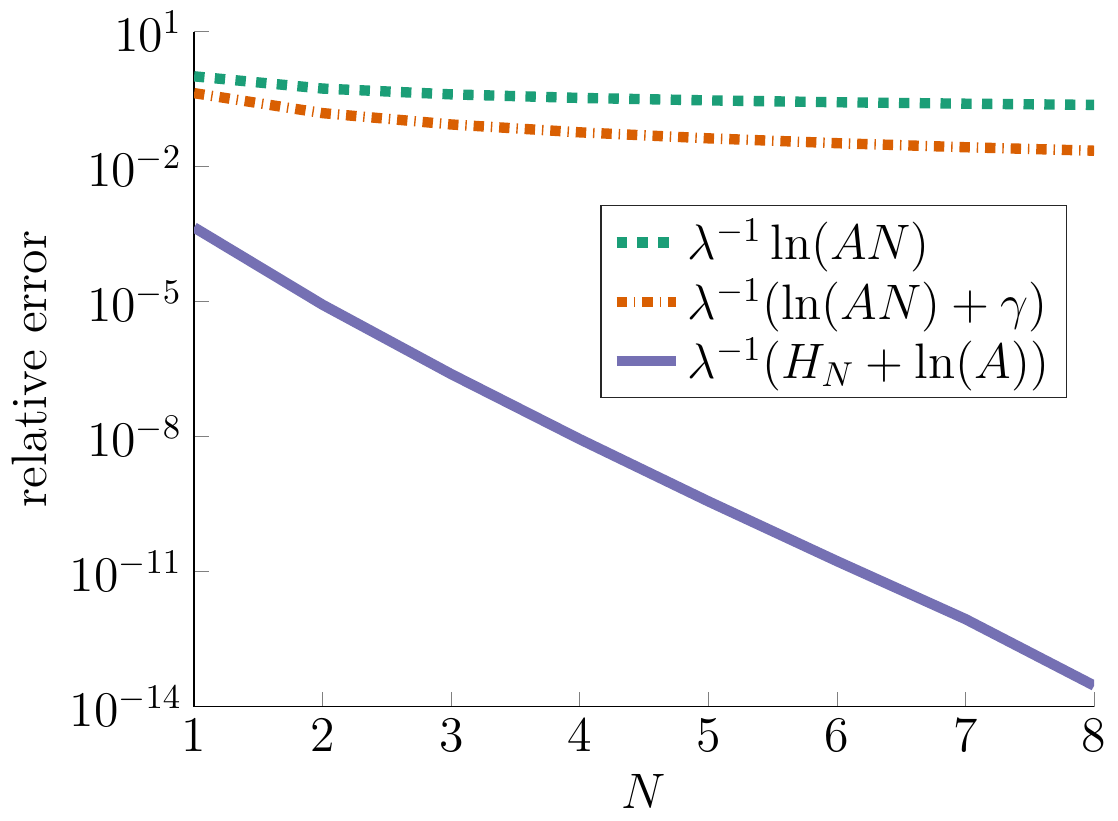}
 \caption{Comparison of Theorem~\ref{expmoments} and Theorem~\ref{work} to a problem studied by \cite{grebenkov2022}. See the text for details.
 }
 \label{figkappa}
\end{figure}

\subsection{Random walk on discrete network}

We now briefly describe how to apply our results to a random walk on a discrete network. Let $\{X(t)\}_{t\ge0}$ be an irreducible, continuous-time Markov chain on a finite state space $I$ (i.e.\ the network) with infinitesimal generator matrix $Q\in\R^{|I|\times|I|}$ \cite{norris1998}. Recall that this means that the entry in row $i$ and column $j$ of $Q$ denotes the rate that $X$ jumps from state $i$ to $j$ if $i\neq j$ and the diagonal entries of $Q$ are chosen so that {$Q$} has zero row sums. Let $I_{\target}\subset I$ denote some set of ``target'' states, and define the FPT to $I_{\target}$,
\begin{align*}
\tau
:=\inf\{t>0:X(t)\in I_{\target}\}.
\end{align*}
Let $\rho=\{\rho_{i}\}_{i\in I}=\{\P(X(0)=i)\}_{i\in I}\in\R^{|I|\times1}$ denote the initial distribution of $X$ and assume that $\rho_{i}=\P(X(0)=i)=0$ for all $i\in I_{\target}$, which merely means that $X$ cannot start in the target set. The survival probability is then given by
\begin{align}\label{ctmc}
S(t)
:=\P(\tau>t)
=\overline{\rho}^{\top}e^{\overline{Q}t}\mathbf{1},
\end{align}
where 
{$
\overline{\rho}=\{\rho_{i}\}_{i\in I\backslash I_{\target}}\in\R^{|I\backslash I_{\target}|\times1},
$}
is the vector obtained by discarding the elements of $\rho$ corresponding to states in $I_{\target}$ (and $\overline{\rho}^{\top}\in\R^{1\times|I\backslash I_{\target}|}$ denotes the transpose of $\overline{\rho}$), 
{$
\overline{Q}=\{Q_{i,j}\}_{i,j\in I\backslash I_{\target}}\in\R^{|I\backslash I_{\target}|\times|I\backslash I_{\target}|}
$} 
denotes the matrix obtained by discarding all the rows and columns corresponding to states in $I_{\target}$, and $\mathbf{1}\in\R^{|I\backslash I_{\target}|\times1}$ is the column vector of all ones.

The form of the survival probability in \eqref{ctmc} implies that it must decay at large time according to 
{$
S(t)
\sim A(\lambda t)^{r}e^{-\lambda t}\quad\text{as }t\to\infty,
$} 
where $A>0$, $r\in\{0,1,2,\dots\}$, and $\lambda>0$ depend on $\overline{\rho}$ and $e^{\overline{Q}t}$. Hence, the slowest FPTs satisfy the assumptions Theorems~\ref{expdist}-\ref{expmoments}, where the scalings $b_{N}$ involve either logarithms (if $r=0$) or the lower branch of the {Lambert W} function (if $r\ge1$). Hence, Theorems~\ref{expdist}-\ref{expmoments} (and Theorem~\ref{work} if $r=0$) yield the full distribution and all the moments of $T_{N-k,N}$ if ${{N\gg k+1\ge1}}$.

\subsection{Diffusion on half-line}\label{halfline}

We now consider an example in which the survival probability of a single FPT has power law decay rather than exponential decay. Let $\{X(t)\}_{t\ge0}$ be a one-dimensional, pure diffusion process with diffusivity $D>0$. Let $\tau$ be the FPT to reach the origin,
\begin{align}\label{tauhalf}
\tau
:=\inf\{t>0:X(t)=0\},
\end{align}
and assume that the searcher starts at $X(0)={x}>0$. The important distinction between this example and the diffusion examples above is that the domain in this example is unbounded. 

The survival probability is \cite{carslaw1959}
\begin{align}\label{erf}
S(t)
:=\P(\tau>t)
=\textup{erf}({x}/\sqrt{4Dt}),
\end{align}
where $\textup{erf}(z)=\frac{2}{\sqrt{\pi}}\int_{0}^{z}e^{-u^{2}}\,\dd u$ denotes the error function. Taking $t\to\infty$ in \eqref{erf} yields 
\begin{align}\label{Shalf}
S(t)
\sim(\lambda t)^{-p}\quad\text{as }t\to\infty,
\end{align}
where
\begin{align}\label{Shalfp}
\lambda
=\frac{\pi D}{{x}^{2}},\quad
p
=\frac{1}{2}.
\end{align} 

\begin{figure}
  \centering
             \includegraphics[width=1\textwidth]{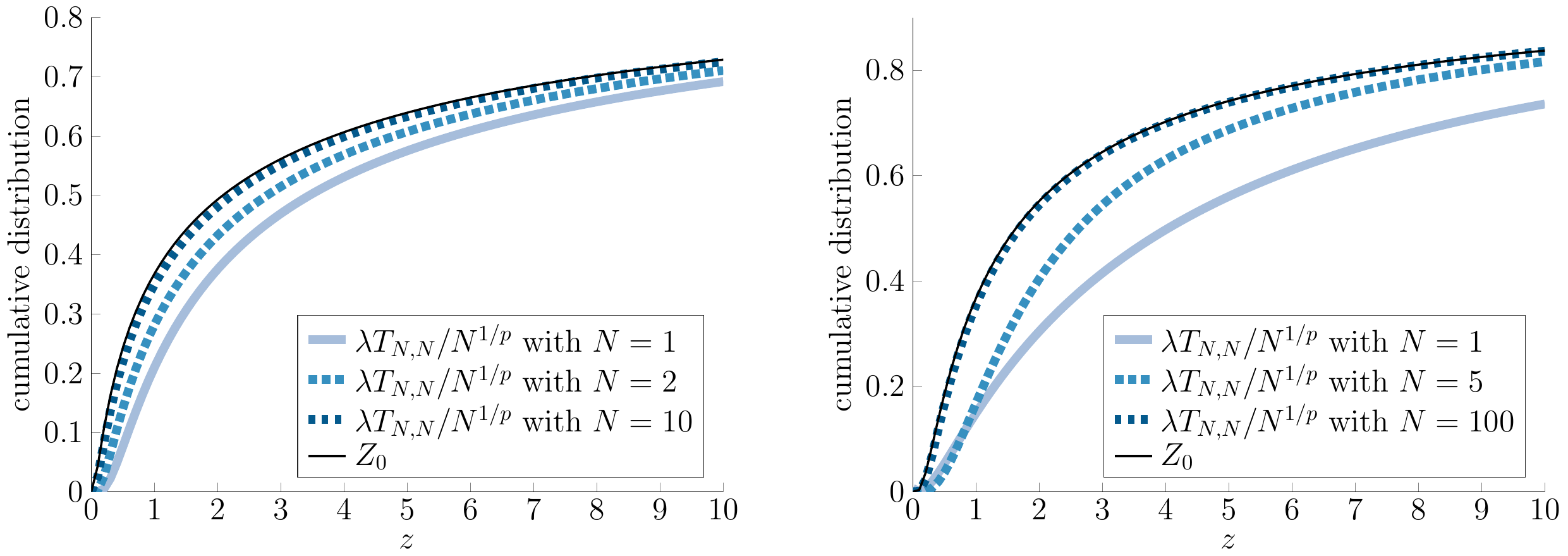}
 \caption{\textbf{Left:} Convergence in distribution for the half-line problem in section~\ref{halfline}. \textbf{Right:} Convergence in distribution for the subdiffusion problem in section~\ref{sub}. See the text for details.}
 \label{fighalfsub}
\end{figure}

In the left panel of Figure~\ref{fighalfsub}, we plot the convergence in distribution implied by Theorem~\ref{powerdist}. We again see that the convergence rate is rapid for this slowest FPT. In this plot, we take $D=x=1$.

Since $p=1/2$, notice that Theorem~\ref{powermoments} and \eqref{cii} implies that
\begin{align*}
\E[T_{N-k,N}]
<\E[T_{N,N}]
=\E[\tau]=\infty,\quad\text{for any }2<k+1\le N.
\end{align*}
That is, the third slowest searcher out of $N\ge3$ searchers has a finite mean FPT, despite the fact that any given searcher has an infinite mean FPT. This result is counterintuitive in the case of many searchers, since it means that the third slowest out of $N\gg1$ searchers is actually faster than a typical searcher, in the sense that $\E[T_{N-2,N}]<\infty$ and $\E[\tau]=\infty$.

Generalizing this example, suppose each diffusing searcher experiences a constant drift toward the origin. That is, the position of a searcher evolves according to the SDE,
\begin{align}\label{drift}
\dd X
=-V\,\dd t+\sqrt{2D}\,\dd W,\quad X(0)=x>0,
\end{align}
where $W$ is a standard Brownian motion and the drift $V>0$ pushes the searcher ``down'' toward the origin. We show in section~\ref{apphalfline} in the Appendix that the survival probability for \eqref{tauhalf} is
\begin{align}\label{SV}
S(t)=\frac{1}{2}\bigg[1+\textup{erf}\Big(\frac{x-Vt}{\sqrt{4Dt}}\Big)
-e^{\frac{Vx}{D}}\textup{erfc}\Big(\frac{x+Vt}{\sqrt{4Dt}}\Big)\bigg],\quad t>0.
\end{align}
Expanding \eqref{SV} as $t\to\infty$ yields
\begin{align}\label{ddecay}
S(t)
\sim A(\lambda t)^{-p}e^{-\lambda t}\quad\text{as }t\to\infty,
\end{align}
where 
\begin{align*}
\lambda
=\frac{V^{2}}{4D},\quad
A
=\frac{Vx}{4D\sqrt{\pi}}e^{Vx/(2D)},\quad
p
=\frac{3}{2}.
\end{align*}
Hence, the drift causes the survival probability to decay exponentially rather than according to the power law in \eqref{Shalf}. Further, Theorem~\ref{expdist} and the power law prefactor in \eqref{ddecay} imply that the distributions and moments of the slowest FPTs for this example are described in terms of the {Lambert W} function. We postpone numerical illustrations of this example until section~\ref{mortal} below, where we show that the presence of the drift in \eqref{drift} is exactly equivalent to considering purely diffusive searchers (i.e.\ with no drift) that are conditioned to find the target before an exponentially distributed inactivation time.

\subsection{Subdiffusive search}\label{sub}

Subdiffusive stochastic motion has been observed in a variety of diverse physical scenarios and is especially prevalent in cell biology \cite{oliveira2019, klafter2005, barkai2012, sokolov2012}. While diffusion is marked by a mean-squared displacement that grows linearly in time, anomalous subdiffusion is defined by a mean squared displacement that grows like $t^{\alpha}$ as time $t$ increases, where $\alpha\in(0,1)$ is the subdiffusive exponent. One very common way to model subdiffusion is via a fractional Fokker-Planck equation \cite{metzler1999} (which can be derived from the continuous-time random walk model with power law waiting times \cite{metzler2000}), which is equivalent to constructing a subdiffusive process $\{X(t)\}_{t\ge0}$ via a random time change of a diffusive process $\{Y(s)\}_{s\ge0}$ \cite{lawley2020subpre}. More specifically, if $\{Y(s)\}_{s\ge0}$ is a diffusion process satisfying an It\^{o} stochastic differential equation, then the subdiffusive process is defined via
\begin{align}\label{XYS}
X(t)
:=Y(\mathbb{S}(t)),\quad t\ge0,
\end{align}
where $\{\mathbb{S}(t)\}_{t\ge0}$ is an inverse $\alpha$-stable subordinator that is independent of $Y$.

Therefore, if $\tau$ and $\sigma$ denote the respective FPTs of $X$ and $Y$ to some target, then 
\begin{align}\label{rep}
S(t)
:=\P(\tau>t)
=\E[S_{\sigma}(\mathbb{S}(t))],
\end{align}
where $S_{\sigma}(s):=\P(\sigma>s)$ denotes the survival probability of the diffusive FPT $\sigma$. Using that the probability density that $\mathbb{S}(t)=s$ is given by
\begin{align}\label{lspdf}
\frac{\dd}{\dd s}\P(\mathbb{S}(t)\le s)
=\frac{t}{\alpha s^{1+1/\alpha}}l_{\alpha}\Big(\frac{t}{s^{1/\alpha}}\Big),
\end{align}
where $l_{\alpha}(z)$ is defined via its Laplace transform,
\begin{align}\label{ltl}
\int_{0}^{\infty}e^{-rz}l_{\alpha}(z)\,\dd z
=e^{-r^{\alpha}},\quad \alpha\in(0,1),\; r\ge0,
\end{align}
the representation \eqref{rep} yields
\begin{align}\label{reprep}
S(t)
=\int_{0}^{\infty}\P(\sigma>s)\frac{t}{\alpha s^{1+1/\alpha}}l_{\alpha}\Big(\frac{t}{s^{1/\alpha}}\Big)\,\dd s.
\end{align}

The following proposition is a general result that yields the large time behavior of any function $S(t)$ satisfying \eqref{reprep} assuming that $\sigma>0$ has finite mean. The proof is given in section~\ref{appsub} in the Appendix.

\begin{proposition}\label{thmsub}
Let $\sigma>0$ be any nonnegative random variable with finite mean. If $S(t)$ is given by \eqref{reprep} where $l_{\alpha}$ is defined via \eqref{ltl}, then
\begin{align}\label{subeq}
S(t)
\sim\frac{\E[\sigma]}{\Gamma(1-\alpha)}t^{-\alpha}\quad\text{as }t\to\infty.
\end{align}
\end{proposition}

Applying Proposition~\ref{thmsub} to the case of subdiffusion described above, we obtain the large-time behavior of the survival probability of a subdiffusive FPT (we note the large-time decay in \eqref{subeq} was derived formally in \cite{condamin2007first} and \cite{lua2005} under stronger assumptions). Combining this result with Theorem~\ref{powerdist} yields the probability distribution for the slowest subdiffusive FPTs. Specifically, if $\{\tau_{n}\}_{n\ge1}$ denote iid subdiffusive FPTs whose survival probabilities satisfy \eqref{subeq}, then Theorem~\ref{powerdist} implies that
\begin{align}\label{cdsub}
\frac{\lambda T_{N-k,N}}{N^{1/p}}
\to_{\dd}Z_{k}\quad\text{as }N\to\infty,
\end{align}
where 
\begin{align*}
p=\alpha,\qquad
\lambda
=\Big(\frac{\Gamma(1-\alpha)}{\E[\sigma]}\Big)^{1/\alpha},\qquad
\P(Z_{k}\le {z})
=\frac{\Gamma(k+1,z^{-\alpha})}{k!},\quad\text{if }z>0.
\end{align*}
As expected, this shows that slowest subdiffusive FPTs are much slower than slowest diffusive FPTs. This intuitive result contrasts the results of \cite{lawley2020sub}, wherein it was proven that the fastest subdiffusive searchers are faster than the fastest diffusive searchers. We also note that Theorem~\ref{powermoments} implies $\E[T_{N,N}]=\E[\tau]=\infty$ and $\E[T_{N-k,N}]<\infty$ if $1<(k+1)\alpha\le N\alpha$.

The convergence in distribution in \eqref{cdsub} is illustrated in the right panel of Figure~\ref{fighalfsub}. For this plot, $X$ is in \eqref{XYS} where $Y$ is a one-dimensional pure diffusion process with diffusivity $D>0$ starting at the origin, and the diffusive FPT $\sigma$ is the first time that $Y$ escapes the interval $(-L,L)$. We take $D=1$, $L=1/2$, and $\alpha=3/4$ for the subdiffusive process $X$. Details on the numerical methods for this example are in section~\ref{appsub} in the Appendix.

\section{Mortal searchers and signal sharpness}\label{mortal}

We now consider so-called ``mortal'' searchers, which may ``die'' (degrade, be inactivated, etc.)\ before finding the target. Such search processes are sometimes called ``evanescent'' and have been studied extensively \cite{abad2010, abad2012, abad2013, yuste2013, meerson2015b, grebenkov2017, ma2020, lawley2021mortal}.  

\subsection{Large-time survival probability asymptotics}

Mathematically, the FPT $\tau^{\textup{mortal}}$ of such a mortal searcher can be written as
\begin{align*}
\tau^{\textup{mortal}}
:=\begin{cases}
\tau & \text{if }\tau\le\sigma,\\
+\infty & \text{if }\tau>\sigma,
\end{cases}
\end{align*}
where $\tau$ is the FPT of an immortal searcher (i.e.\ a searcher with no inactivation) and $\sigma>0$ denotes the inactivation time. Following typical assumptions, we assume that $\sigma$ is independent of $\tau$ and is exponentially distributed with mean $\E[\sigma]=1/r$. 

Let $\overline{\tau}$ denote the FPT of a mortal searcher that is conditioned to find the target before inactivation. That is, $\overline{\tau}$ has survival probability,
\begin{align}\label{taubar}
\overline{S}(t)
:=\P(\overline{\tau}>t)
:=\P(\tau>t\,|\,\tau\le\sigma)
=\frac{\P(t<\tau\le\sigma)}{\P(\tau\le\sigma)}.
\end{align}
Such inactivation has the effect of filtering out searchers which take a long time to find the target. The following simple result computes the large-time decay of $\overline{S}(t)$ based on the large-time decay of $S(t)$ (the proof is collected in section~\ref{appmortal} in the Appendix).

\begin{proposition}\label{mortalthm}
Assume $\overline{S}(t):=\P(\overline{\tau}>t)$ is in \eqref{taubar}, where $\sigma$ is independent of $\tau$ and exponentially distributed with mean $\E[\sigma]=1/r$ and $S(t):=\P(\tau>t)$. If
\begin{align*}
S(t)-Ae^{-\lambda t}
=O(e^{-\beta t})\quad\text{as }t\to\infty,
\end{align*}
where $0<\lambda<\beta$ and $A>0$, then 
\begin{align*}
\overline{S}(t)-\overline{A}e^{-\overline{\lambda}t}
=O(e^{-(\beta+r) t})\quad\text{as }t\to\infty,
\end{align*}
where
\begin{align*}
\overline{\lambda}
=\lambda+r,\quad
\overline{A}
=\frac{A\lambda/(\lambda+r)}{\int_{0}^{\infty}(1-S(s))r e^{-r s}\,\dd s}.
\end{align*}
If $S(t)\sim (\lambda t)^{-p}$ as $t\to\infty$, where $\lambda>0$ and $p>0$, then
\begin{align}\label{mortalexp}
\overline{S}(t)
\sim \overline{A}(\overline{\lambda} t)^{-\overline{p}}e^{-\overline{\lambda}t}\quad\text{as }t\to\infty,
\end{align}
where
\begin{align*}
\overline{\lambda}
=r,\quad
\overline{p}
=p+1,\quad
\overline{A}
=\frac{p(r/\lambda)^{p}}{\int_{0}^{\infty}(1-S(s))r e^{-r s}\,\dd s}.
\end{align*}
\end{proposition}

Combining Proposition~\ref{mortalthm} with Theorems~\ref{expdist}-\ref{expmoments} yields the distribution and moments of the slowest FPT for mortal searchers that find the target before inactivation.

\subsection{Half-line}\label{hm}

\begin{figure}
  \centering
             \includegraphics[width=1\textwidth]{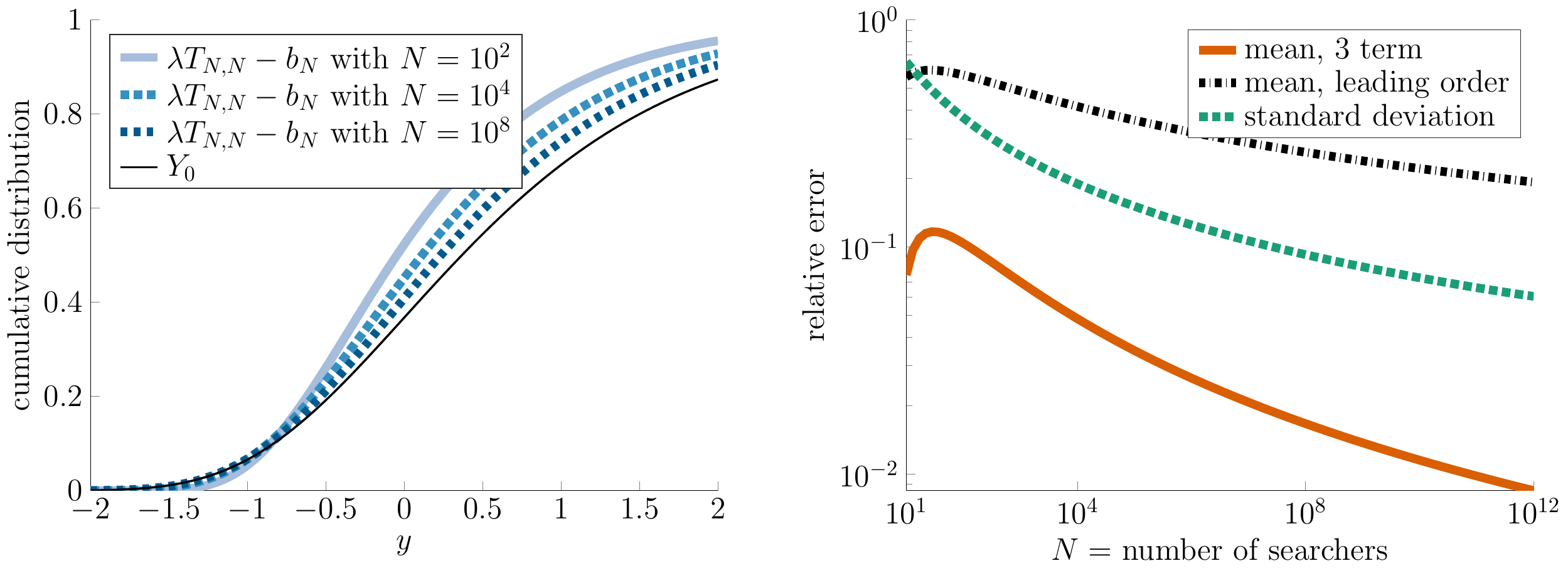}
 \caption{Convergence in distribution and moments for mortal searchers conditioned to find the target before inactivation. The left panel plots the convergence in distribution in \eqref{cdmw}. The right panel plots the relative errors in \eqref{remean}-\eqref{restd} for computing the mean and standard deviation of $T_{N,N}$ via Theorem~\ref{expmoments}. See section~\ref{hm} for details.}
 \label{figmortal}
\end{figure}

Consider the example in section~\ref{halfline} of pure diffusion on the positive real line starting from $X(0)=x>0$ and a target at the origin. The survival probability for a single immortal searcher has the power law decay in \eqref{Shalf}-\eqref{Shalfp}. Using the unconditioned survival probability in \eqref{erf}, Proposition~\ref{mortalthm} implies that the conditioned survival probability $\overline{S}(t)$ in \eqref{taubar} decays exponentially according to \eqref{mortalexp} with
\begin{align*}
\overline{A}
=\frac{1}{\sqrt{4\pi}}e^{\sqrt{rx^{2}/D}}\sqrt{rx^{2}/D},\quad
\overline{\lambda}
=r,\quad
\overline{p}
=\frac{3}{2}.
\end{align*}
Hence, Theorems~\ref{expdist}-\ref{expmoments} yield the distribution and moments of the slowest FPTs in terms of the {Lambert W} function. In particular, if $T_{N-k,N}$ is as in \eqref{Tkn} but with $\tau$ replaced by $\overline{\tau}$, then Theorem~\ref{expdist} implies
\begin{align}\label{cdmw}
\overline{\lambda}T_{N-k,N}-\overline{p}W_{0}\big((\overline{A}N)^{1/\overline{p}}/\overline{p}\big)
\to_{\dd} Y_{k}\quad\text{as }N\to\infty.
\end{align}
Further, Theorem~\ref{expmoments} yield the three-term asymptotic expansion for the mean,
\begin{align}\label{3term}
\E[T_{N-k,N}]
=(1/\overline{\lambda})\Big(\ln N-\overline{p}\ln\ln N+\ln \overline{A}+\gamma-H_{k}+o(1)\Big),\quad\text{as }N\to\infty.
\end{align}

To illustrate these results numerically, we first note that a direct calculation using \eqref{erf} and \eqref{taubar} shows that for this example,
\begin{align}\label{barV}
\overline{S}(t)
=\frac{1}{2}\bigg[1+\textup{erf}\Big(\frac{x-t\sqrt{4Dr}}{\sqrt{4Dt}}\Big)
-e^{\frac{x\sqrt{4Dr}}{D}}\textup{erfc}\Big(\frac{x+t\sqrt{4Dr}}{\sqrt{4Dt}}\Big)\bigg].
\end{align}
Notice that \eqref{barV} is exactly equivalent to \eqref{SV} if the drift in \eqref{SV} is given by
\begin{align}\label{Vr}
V
=\sqrt{4Dr}>0.
\end{align}

In Figure~\ref{figmortal}, we investigate the distribution and moments of $T_{N-k,N}$ for this example. Due to the equivalence of \eqref{barV} and \eqref{SV} if \eqref{Vr} holds, Figure~\ref{figmortal} also applies to the example with drift in section~\ref{halfline}. In the left panel of Figure~\ref{figmortal}, we plot the distribution of $\lambda T_{N,N}-b_{N}$. While the convergence in distribution to $Y_{0}$ is evident, the rate of convergence is markedly slower than in the examples considered above. This slow convergence is also seen in the right panel of Figure~\ref{figmortal}, where we plot the relative errors for the three-term estimate of the mean of $T_{N,N}$ in \eqref{3term} (solid orange curve) and the estimate of the standard deviation given by Theorem~\ref{expmoments} (dashed green curve). We also plot the relative error for the mean if we only {used} the leading order logarithmic term in \eqref{3term} (dot dashed black curve), which is much larger than using the full three-term estimate in \eqref{3term}. Hence, this more detailed three-term estimate in \eqref{3term} is necessary for an accurate estimate of the mean, with the iterated logarithmic term making a strong contribution. In this plot, we take $D=x=r=1$.

\subsection{Signal sharpness}

The relevance of mortal searchers to cell signaling was recently highlighted in the very interesting study by \cite{ma2020}, wherein the authors showed that inactivation ``sharpens'' signals by reducing variability in FPTs. \cite{ma2020} were particularly interested in a signal transmitted by diffusing proteins (the searchers) which move from the cell membrane to the nucleus (the target). For a signal conveyed by a finite number of searchers, signal ``sharpness'' could be understood as inversely related to the signal ``spread,'' defined as the difference between the latest and earliest searchers to arrive at the target. That is, a notion of signal spread is
\begin{align}\label{ss88}
\Sigma_{N}
:=T_{N,N}-T_{1,N}>0,
\end{align}
where $T_{N,N}$ and $T_{1,N}$ are the respective slowest and fastest searchers to arrive at the target. We now use the results above to investigate how inactivation (i.e.\ mortal searchers) affects the signal spread $\Sigma_{N}$. 

A typical model of cell signaling, such as the model in \cite{ma2020}, involves diffusive searchers in a bounded domain with reflecting boundaries. In such a model, the survival probability of an immortal searcher can be written as a sum of decaying exponentials,
\begin{align*}
S(t)
=Ae^{-\lambda t}+\sum_{n\ge1}B_{n}e^{-\beta_{n}t},
\end{align*}
where $0<\lambda<\beta_{1}<\cdots$. Suppose that searchers are inactivated at rate $r>0$, and consider $\Sigma_{N}$ in \eqref{ss88} for searchers which reach the target before inactivation. Applying Proposition~\ref{mortalthm} and Theorem~\ref{work} to $T_{N,N}$ and the results of \cite{lawley2020uni} to $T_{1,N}$ yields the following large $N$ behavior of the mean of $\Sigma_{N}$,
\begin{align}\label{ss87}
\E[\Sigma_{N}]
=\frac{1}{\lambda+r}\Big(\ln N+\ln \overline{A}+\gamma\Big)
-\frac{L^{2}}{4D}\frac{1}{\ln N}
+o(1/\ln N)\quad\text{as }N\to\infty,
\end{align}
where $L>0$ is the shortest distance from the searcher starting location to the target and $D>0$ is the searcher diffusivity,

In this model, the three important timescales are
\begin{align*}
L^{2}/D,\quad
1/\lambda,\quad
1/r.
\end{align*}
As in section~\ref{rare}, it is often the case that
\begin{align}\label{typ88}
L^{2}/D
\ll 1/\lambda,
\end{align}
which means that an immortal searcher tends to wander around the domain before finding the target. That is, $L^{2}/D$ describes the FPT of searchers which move along the shortest path to the target, which is often much faster than $1/\lambda$.

\begin{figure}
  \centering
             \includegraphics[width=.6\textwidth]{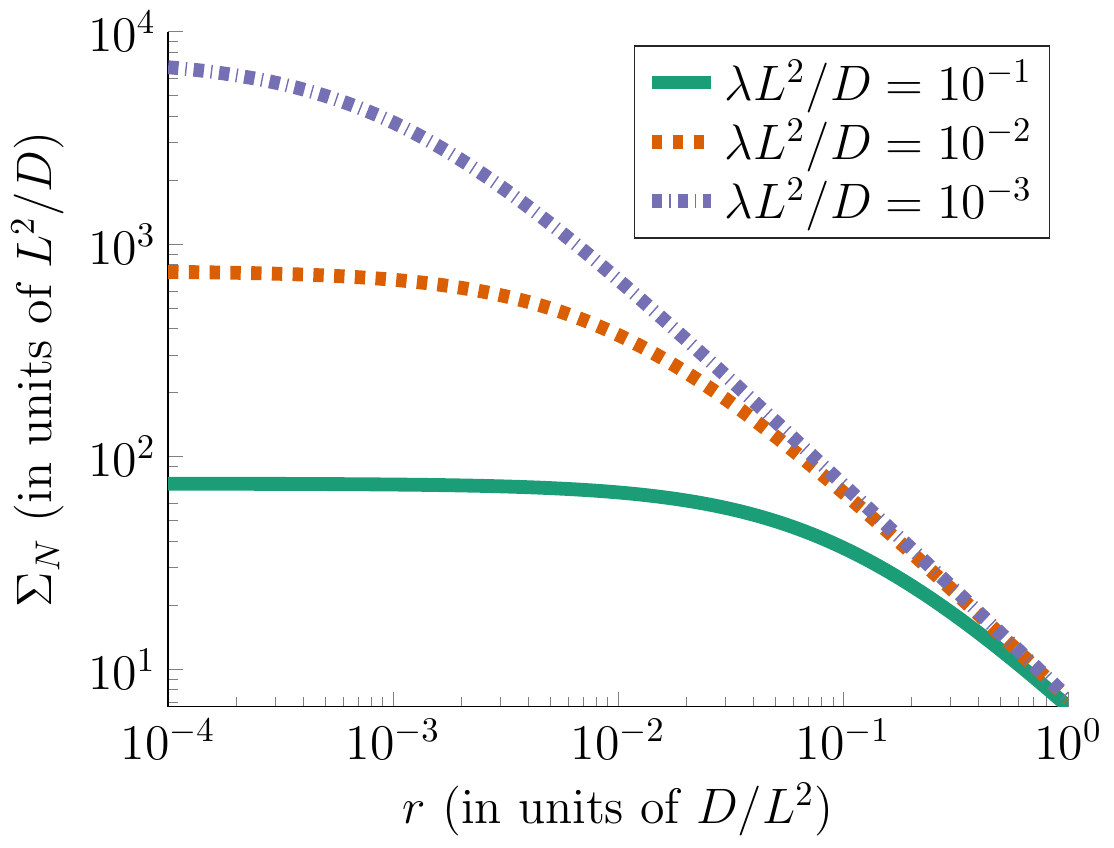}
 \caption{Signal spread $\Sigma_{N}$ in \eqref{ss87} for mortal searchers. See the text for details.}
 \label{figsharpen}
\end{figure}

In Figure~\ref{figsharpen}, we plot the expression for $\E[\Sigma_{N}]$ in \eqref{ss87} ignoring the $o(1/\ln N)$ term. In this plot, we take $\lambda L^{2}/D\in\{10^{-1},10^{-2},10^{-3}\}\ll1$ and consider the inactivation rate $r$ ranging from $r=0$ (i.e.\ immortal searchers) to $r=D/L^{2}$. We take $N=10^{3}$ and $\ln\overline{A}=0$ for simplicity. Notice that there is a drastic decrease in $\Sigma_{N}$ once $r$ is larger than $\lambda$. That is, if $r\gg\lambda$, then $\Sigma_{N}$ is much less than the value of $\Sigma_{N}$ {for} immortal searchers, even if $r\ll D/L^{2}$. This means that inactivation strongly sharpens the signal as long as the inactivation rate is sufficiently large that it filters out searchers which wander around the entire domain, without requiring the inactivation rate to be so large that the only searchers which find the target are those that move along the shortest path to the target. 
\section{Discussion}\label{secdiscussion}

In this paper, we obtained rigorous approximations to the distribution and moments of slowest FPTs. The mathematical results relied on extreme value theory and detailed asymptotic estimates. {We emphasize that these general results apply to the largest values of any sequence of iid nonnegative random variables whose large-time distribution decays either (a) exponentially (possibly with a power law pre-factor) or (b) according to a power law. We applied these general results to FPTs of stochastic searchers. We mostly considered searchers move by pure diffusion, though our results apply to FPTs of much more general search processes, including diffusion with space-dependent drift and noise coefficients, subdiffusive motion, continuous-time Markov chains, non-Markovian random walks \cite{levernier2022}, and search processes with stochastic initial positions.}

{The general} results were proven in the limit of many searchers, but numerical simulations demonstrated their high accuracy for any number of searchers in some typical scenarios of interest. This contrasts with existing estimates of fastest FPTs, which generally require a very large number of searchers to be accurate. This study was motivated by diverse biological systems, including ovarian aging, cell signaling, and single-cell source location detection.

As described in the Introduction, the seemingly redundant excesses in various biological systems have been understood as a means to accelerate search processes \cite{schuss2019}. The oversupply of male gametes (i.e.\ sperm cells) in human fertilization constitutes a prototypical example \cite{meerson2015}. In particular, prior works have argued that the excess of searchers in some systems serves to accelerate the fastest FPT. The analysis in this paper suggests that the excess female gametes (i.e.\ PFs) present at birth ensure a supply available for ovulation for several decades of life. In particular, we have argued that the excess of searchers in this system serves to prolong the slowest FPT. It would be interesting to investigate how this principle might operate in other biological systems, wherein many redundant copies ensure that the supply lasts much longer than the typical lifespan of any single copy.

In the case of our random walk model of the human ovary, the three orders of magnitude oversupply around the time of birth can be seen to all but ensure that the supply {of} follicles exceeds 40 years.  Because human egg quality is known to decline late in a woman's 30's across the population \cite{pmid9922101,pmid19589949}, ovarian function will thus almost always function for longer than the supply of high quality eggs capable of conception. In addition, because the functioning ovary is well-known to support key health and well-being measures in women that tend to decline at the time of menopause \cite{pmid11791087,pmid18310277,pmid28703650}, slowest FPTs can also be seen to dictate the timing that these important life changes take place.

We conclude by discussing how slowest FPTs can be nearly deterministic. Indeed, Theorem~\ref{expmoments} implies that the coefficient of variation of $T_{N-k,N}$ vanishes as $N\to\infty$ (see \eqref{cvvanish}). For the model of ovarian aging in section~\ref{meno}, Theorem~\ref{expmoments} implies that for a given woman with PF starting supply $N\in[10^{5},10^{6}]$, the standard deviation for her age at menopause is less than two months. That is, even though each PF leaves the reserve at a random time, the time of PF exhaustion is quite predictable from the PF starting supply. This predictable behavior is a consequence of the large number of searchers, but we emphasize that the mechanism is quite different from the classical law of large numbers. The law of large numbers relies on averaging many stochastic realizations, whereas the nearly deterministic behavior of slowest FPTs stems from rare events. Indeed, as often noted in large deviation theory \cite{den2000}, rare events are predictable in that they occur in the least unlikely of all the unlikely ways.

\subsubsection*{Acknowledgements}
The authors gratefully acknowledge very helpful input from John W. Emerson (Yale University).

\appendix
\section{Appendix}

\subsection{Proofs of theorems}\label{appproofs}

In this section, we collect the proofs of the theorems in section~\ref{math}. Before proving Theorem~\ref{expdist}, we first state and prove a simple lemma.

\begin{lemma}\label{simplelemma}
Suppose $\{s_{N}\}_{N\ge1}$ and $\{s_{N}'\}_{N\ge1}$ are sequences of real numbers satisfying
\begin{align}
s_{N}
&\sim s_{N}'\quad\text{as }N\to\infty,\label{equivi}\\
(1-s_{N}')^{N}
&\to z\in(0,\infty)\quad\text{as }N\to\infty,\label{wanti}
\end{align} for some $z\in(0,\infty)$. Then
\begin{align}\label{wantnoprime}
(1-s_{N})^{N}
&\to z\in(0,\infty)\quad\text{as }N\to\infty.
\end{align}
\end{lemma}

\begin{proof}[Proof of Lemma~\ref{simplelemma}]
Since the logarithm and exponential functions are continuous, \eqref{wanti} is equivalent to
\begin{align}\label{wantii}
\lim_{N\to\infty}N\ln (1-s_{N}')
=\ln z.
\end{align}
Since \eqref{wantii} implies that $s_{N}'\to0$ as $N\to\infty$, applying L'H\^{o}pital's rule yields $\ln(1-s_{N}')\sim -s_{N}'$ as $N\to\infty$. Therefore, \eqref{wantii} is equivalent to
\begin{align}\label{wantiii}
\lim_{N\to\infty}Ns_{N}'
=-\ln z.
\end{align}
By \eqref{equivi}, we have that \eqref{wantiii} holds with $s_{N}'$ replaced by $s_{N}$. But, by the same argument that yielded \eqref{wantiii} from \eqref{wanti}, we have that \eqref{wantnoprime} holds.
\end{proof}

\begin{proof}[Proof of Theorem~\ref{expdist}]
We first take $b_{N}=b_{N}'$ in \eqref{bN}. Fix $y\in\R$. Let $S_{0}(t):=A(\lambda t)^{-p}e^{-\lambda t}$. Since $b_{N}\to\infty$ as $N\to\infty$, we have that
\begin{align}\label{nice87}
S_{0}(\lambda^{-1}(y+b_{N}))
=A(y+b_{N})^{-p}e^{-b_{N}}e^{-y}
\sim Ab_{N}^{-p}e^{-b_{N}}e^{-y}
\quad\text{as }N\to\infty.
\end{align}
Furthermore, the definition of $b_{N}$ ensures that $Ab_{N}^{-p}e^{-b_{N}}=N^{-1}$, and therefore \eqref{nice87} implies
\begin{align*}
S_{0}(\lambda^{-1}(y+b_{N}))
\sim N^{-1}e^{-y}
\quad\text{as }N\to\infty.
\end{align*}
Since $\lim_{N\to\infty}(1-N^{-1}e^{-y})^{N}=\exp(-e^{-y})$, Lemma~\ref{simplelemma} yields
\begin{align}\label{want1}
\lim_{N\to\infty}(1-S_{0}(\lambda^{-1}(y+b_{N})))^{N}
=\exp(-e^{-{y}}),\quad \text{for all }{y}\in\R.
\end{align}

By assumption, $S(t)\sim S_{0}(t)$ as $t\to\infty$, and therefore $S(\lambda^{-1}(y+b_{N}))\sim S_{0}(\lambda^{-1}(y+b_{N}))$ as $N\to\infty$ since $b_{N}\to\infty$ as $N\to\infty$. Hence, \eqref{want1} and Lemma~\ref{simplelemma} imply
\begin{align}\label{want2}
\lim_{N\to\infty}(1-S(\lambda^{-1}(y+b_{N})))^{N}
=\exp(-e^{-{y}}),\quad \text{for all }{y}\in\R.
\end{align}
Therefore, since $T_{N,N}=\max\{\tau_{1},\dots,\tau_{N}\}$ and $\{\tau_{n}\}_{n\ge1}$ are iid, we have
\begin{align}\label{triv}
\begin{split}
\P(\lambda T_{N,N}-b_{N}\le y)
=\P(T_{N,N}\le \lambda^{-1}(y+b_{N}))
&=\big(\P(\tau_{1}\le \lambda^{-1}(y+b_{N}))\big)^{N}\\
&=(1-S(\lambda^{-1}(y+b_{N})))^{N}.
\end{split}
\end{align}
Taking $N\to\infty$ in \eqref{triv} and using \eqref{want2} yields that $\lambda T_{N,N}-b_{N}\to_{\dd} Y_{k}$, which completes the proof for the case $k=0$. Having established the case $k=0$, the case of a general fixed integer $k\ge1$ follows directly from Theorem 3.4 in \cite{coles2001}. The fact that $b_{N}$ can be replaced by any sequence satisfying \eqref{bnp} was shown in \cite{gnedenko1943} (see \cite{peng2012} for a more recent reference).
\end{proof}

\begin{proof}[Proof of Theorem~\ref{expmoments}]
Since Theorem~\ref{expdist} establishes that $\lambda T_{N-k,N}-b_{N}$ converges in distribution to $Y_{k}$ as $N\to\infty$, the continuous mapping theorem (see, for example, Theorem~2.7 in \cite{billingsley2013}) implies that $(\lambda T_{N-k,N}-b_{N})^{m}$ converges in distribution to $(Y_{k})^{m}$ as $N\to\infty$. To conclude $\E[(\lambda T_{N-k,N}-b_{N})^{m}]\to\E[(Y_{k})^{m}]$ as $N\to\infty$, it is enough to show that the sequence of random variables $\{(\lambda T_{N-k,N}-b_{N})^{m}\}_{N\ge1}$ is uniformly integrable (see, for example, Theorem~3.5 in \cite{billingsley2013}). To show this uniform integrability, it is enough (see, for example, equation~(3.18) in \cite{billingsley2013}) to show that
\begin{align*}
\sup_{N\ge1}\E\Big[\big(\lambda T_{N-k,N}-b_{N}\big)^{r}\Big]<\infty,
\end{align*}
where $r=2m\ge2$ is an even integer.

Now, 
\begin{align}\label{terms4}
\E[(\lambda T_{N-k,N}-b_{N})^{r}]
=\E[(\lambda T_{N-k,N}-b_{N})_{+}^{r}]
+\E[(b_{N}-\lambda T_{N-k,N})_{+}^{r}],
\end{align}
where $(\cdot)_{+}$ denotes the positive part (i.e.\ $(x)_{+}=x$ if $x>0$ and $(x)_{+}=0$ otherwise). For the case $k=0$, Theorem~2.1 in \cite{pickands1968} implies
\begin{align*}
\lim_{N\to\infty}\E[((\lambda T_{N,N}-b_{N})_{+})^{r}]
=\int_{0}^{\infty}x^{m}e^{-x-e^{-x}}\,\dd x<\infty.
\end{align*}
Since $T_{N-k,N}\le T_{N,N}$ for any $k\in\{0,1,\dots,N-1\}$, we thus have that
\begin{align*}
\sup_{N\ge1}\E[((\lambda T_{N-k,N}-b_{N})_{+})^{r}]
\le\sup_{N\ge1}\E[((\lambda T_{N,N}-b_{N})_{+})^{r}]<\infty.
\end{align*}

Since $\E[Z]=\int_{0}^{\infty}\P(Z>z)\,\dd z$ for any nonnegative random variable, the second term in the righthand side of \eqref{terms4} can be written as
\begin{align*}
\E[((b_{N}-\lambda T_{N-k,N})_{+})^{r}]
&=\int_{0}^{(b_{N})^{r}}\P\Big(T_{N-k,N}<\frac{b_{N}-s^{1/r}}{\lambda}\Big)\,\dd s,
\end{align*}
since $\tau>0$ almost surely. Now,
\begin{align*}
\P(T_{N-k,N}< t)
&=\P(T_{N,N}< t)+\sum_{j=1}^{k}\P(T_{N-j,N}< t\le T_{N-j+1,N})\\
&=\sum_{j=0}^{k}{N\choose j}[G(t)]^{N-j}[1-G(t)]^{j},
\end{align*}
where $G(t)=\P(\tau<t)$. Therefore,
\begin{align}\label{I2sum}
&\E[((b_{N}-\lambda T_{N-k,N})_{+})^{r}]\nonumber\\
&\quad=\sum_{j=0}^{k}{N\choose j}\int_{0}^{(b_{N})^{r}}\Big[G\Big(\frac{b_{N}-s^{1/r}}{\lambda}\Big)\Big]^{N-j}\Big[1-G\Big(\frac{b_{N}-s^{1/r}}{\lambda}\Big)\Big]^{j}\,\dd s\\
&\quad=:\sum_{j=0}^{k}{N\choose j}I_{j},\nonumber
\end{align}
where $I_{j}$ is the integral in the $j$th term in \eqref{I2sum}. Since ${N\choose j}=O(N^{j})$ as $N\to\infty$, it remains to show that
\begin{align}\label{wtsif}
I_{j}
=O(N^{-j})\quad\text{as }N\to\infty.\end{align}

To show \eqref{wtsif} for $j=0$, observe that
\begin{align*}
\int_{0}^{(b_{N})^{r}}\Big[G\Big(\frac{b_{N}-s^{1/r}}{\lambda}\Big)\Big]^{N}\,\dd s
&=\int_{0}^{(b_{N})^{r}}\P\Big(T_{N,N}<\frac{b_{N}-s^{1/r}}{\lambda}\Big)\,\dd s\\
&=\int_{0}^{(b_{N})^{r}}\P\big(b_{N}-\lambda T_{N,N}> s^{1/r})\,\dd s\\
&=\int_{0}^{(b_{N})^{r}}\P\big(((b_{N}-\lambda T_{N,N})_{+})^{r}> s)\,\dd s\\
&=\E[((b_{N}-\lambda T_{N,N})_{+})^{r}].
\end{align*}
Now, Theorem~2.1 in \cite{pickands1968} implies
\begin{align*}
\lim_{N\to\infty}\E[((b_{N}-\lambda T_{N,N})_{+})^{r}]
=\int_{-\infty}^{0}(-x)^{m}e^{-x-e^{-x}}\,\dd x<\infty,
\end{align*}
and thus \eqref{wtsif} holds for $j=0$.

Fix $j\in\{1,\dots,k\}$ and let $\delta>0$. Splitting $I_{j}$ into the lower integral from $s=0$ to $s=(b_{N}-1/\delta)^{r}$ and the upper integral from $s=(b_{N}-1/\delta)^{r}$ to $s=(b_{N})^{r}$
and estimating the upper integral, we have
\begin{align*}
&\int_{(b_{N}-1/\delta)^{r}}^{(b_{N})^{r}}\Big[G\Big(\frac{b_{N}-s^{1/r}}{\lambda}\Big)\Big]^{N-j}\Big[1-G\Big(\frac{b_{N}-s^{1/r}}{\lambda}\Big)\Big]^{j}\,\dd s\\
&\quad\le\Big[G\Big(\frac{1/\delta}{\lambda}\Big)\Big]^{N-j}\int_{(b_{N}-1/\delta)^{r}}^{(b_{N})^{r}}\Big[1-G\Big(\frac{b_{N}-s^{1/r}}{\lambda}\Big)\Big]^{j}\,\dd s\\
&\quad\le\Big[G\Big(\frac{1/\delta}{\lambda}\Big)\Big]^{N-j}(b_{N})^{r}
=o(N^{-j}),\quad\text{as }N\to\infty,
\end{align*}
since the first factor vanishes exponentially fast as $N\to\infty$. We have used that $G$ is nondecreasing and $G(\frac{1/\delta}{\lambda})<1$ since
\begin{align}\label{Gsim}
1-G(t)\sim Ae^{-\lambda t}\quad\text{as }t\to\infty.
\end{align}

By \eqref{Gsim}, we can take $0<\delta\ll1$ sufficiently small so that
\begin{align*}
\frac{A}{2}(\lambda t)^{-p}e^{-\lambda t}
\le 1-G(t)
\le 2A(\lambda t)^{-p}e^{-\lambda t},\quad\text{if }t\ge\frac{1}{\delta\lambda},
\end{align*}
Moving to the lower part of $I_{j}$, we have that
\begin{align*}
I&:=\int_{0}^{(b_{N}-1/\delta)^{r}}\Big[G\Big(\frac{b_{N}-s^{1/r}}{\lambda}\Big)\Big]^{N-j}\Big[1-G\Big(\frac{b_{N}-s^{1/r}}{\lambda}\Big)\Big]^{j}\,\dd s\\
&\quad\le\int_{0}^{(b_{N}-1/\delta)^{r}}\Big[1-\frac{A}{2}(b_{N}-s^{1/r})^{-p}e^{-(b_{N}-s^{1/r})}\Big]^{N-j}\Big[2A(b_{N}-s^{1/r})^{-p}e^{-(b_{N}-s^{1/r})}\Big]^{j}\,\dd s.
\end{align*}
To analyze this integral, we change variables according to
\begin{align}\label{changev}
\begin{split}
u
&=(b_{N}-s^{1/r})^{-p}e^{-(b_{N}-s^{1/r})},\\
\dd u
&=\frac{u}{r}s^{1/r-1}\Big[1+p/(b_{N}-s^{1/r})\Big]\,\dd s.
\end{split}
\end{align}

We first consider the case $p=0$, in which the change of variables in \eqref{changev} yields $s=(b_{N}-\ln(u^{-1}))^{r}$, and thus
\begin{align*}
I
\le r\int_{1/(AN)}^{\delta'}\Big[1-\frac{A}{2}u\Big]^{N}\Big[2A\Big]^{j}u^{j-1}\big(b_{N}-\ln(u^{-1})\big)^{r-1}\frac{\ln(u^{-1})}{1+\ln(u^{-1})}\,\dd u,
\end{align*}
where we have defined $0<\delta':=\delta^{p}e^{-1/\delta}\ll1$. Since $\ln(u^{-1})\to\infty$ as $u\to0+$, we may take $0<\delta\ll1$ sufficiently small so that 
\begin{align*}
0
<\frac{\ln(u^{-1})}{1+\ln(u^{-1})}
\le2,\quad\text{for all }u\in[1/(AN),\delta'],
\end{align*}
and thus it remains to show that
\begin{align*}
\int_{1/(AN)}^{\delta'}\Big[1-\frac{A}{2}u\Big]^{N}u^{j-1}\big(b_{N}-\ln(u^{-1})\big)^{r-1}\,\dd u
=O(N^{-j})\quad\text{as }N\to\infty.
\end{align*}
Changing variables according to $z=ANu$ and using $b_{N}=\ln(AN)$ yields
\begin{align*}
I''
=(AN)^{-j}\int_{1}^{AN\delta'}\Big[1-\frac{z}{2N}\Big]^{N}z^{j-1}(\ln z)^{r-1}\,\dd u.
\end{align*}
It is straightforward to check that
\begin{align}\label{expbound}
(1-y/N)^{N}
\le e^{-y}\quad\text{for all }y\in[0,N).
\end{align}
To obtain \eqref{expbound}, notice that $f(y):=(1-y/N)^{N}$ satisfies $f(0)=1$ and $f'(y)\le-f(y)$ and apply Gronwall's inequality. Therefore, \eqref{expbound} implies
\begin{align*}
I''
&\le(AN)^{-j}\int_{1}^{\infty}e^{-z/2}z^{j-1}(\ln z)^{r-1}\,\dd u\\
&\le(AN)^{-j}\int_{1}^{\infty}e^{-z}z^{j+r-2}\,\dd u
\le(AN)^{-j}(j+r)!.
\end{align*}

Next, suppose $p>0$. In this case, \eqref{changev} implies $s=(b_{N}-pW(p^{-1}u^{-1/p}))^{r}$ and thus
\begin{align*}
I
\le r\int_{1/(AN)}^{\delta'}\Big[1-\frac{A}{2}u\Big]^{N}\Big[2A\Big]^{j}u^{j-1}\big(b_{N}-pW(p^{-1}u^{-1/p})\big)^{r-1}\frac{W(p^{-1}u^{-1/p})}{1+W(p^{-1}u^{-1/p})}\,\dd u,
\end{align*}
where we have used that
\begin{align}\label{deflambert}
Ab_{N}^{-p}e^{-b_{N}}=N^{-1}.
\end{align}
Since $W(p^{-1}u^{-1/p})\to\infty$ as $u\to0+$, taking $\delta$ sufficiently small ensures that
\begin{align*}
0
<\frac{W(p^{-1}u^{-1/p})}{1+W(p^{-1}u^{-1/p})}
\le2,\quad\text{for all }u\in[1/(AN),\delta'].
\end{align*}
It thus remains to estimate
\begin{align*}
I'
:=\int_{1/(AN)}^{\delta'}\Big[1-\frac{A}{2}u\Big]^{N}u^{j-1}\big(b_{N}-pW(p^{-1}u^{-1/p})\big)^{r-1}\,\dd u.
\end{align*}
Expanding the {Lambert W} function and using the definition of $b_{N}$ yields
\begin{align}\begin{split}\label{expandboth}
pW(p^{-1}u^{-1/p})
&=\ln(u^{-1})
-p\ln\ln p^{-1}u^{-1/p}-p\ln p+g_{0}(u),\\
b_{N}
&=\ln(AN)
-p\ln\ln p^{-1}(AN)^{1/p}
-p\ln p
+g_{1}(N),
\end{split}
\end{align}
where
\begin{align}\label{g0g1}
\lim_{u\to0+}g_{0}(u)=\lim_{N\to\infty}g_{1}(N)=0.
\end{align}
Now, it is straightforward to check that
\begin{align}\label{lee}
0\le \ln b-\ln a\le b-a,\quad\text{if }b\ge a\ge1.
\end{align}
To obtain \eqref{lee}, note that if $f(b)=\ln b-\ln a$ and $g(b)=b-a$, then $f(a)=g(a)=0$ and $f'(b)=1/b\le g'(b)=1$ if $b\ge a\ge1$. Therefore, 
\begin{align}\label{loglog}
\ln\ln p^{-1}(AN)^{1/p}-\ln\ln p^{-1}u^{-1/p}
\le\ln p^{-1}(AN)^{1/p}-\ln p^{-1}u^{-1/p},
\end{align}
if $u\in[1/(AN),\delta']$ and $\delta$ is sufficiently small and $N$ is sufficiently large so that $\ln p^{-1}(AN)^{1/p}\ge\ln p^{-1}u^{-1/p}\ge1$. Hence, using \eqref{expandboth}, \eqref{g0g1}, and \eqref{loglog} yields
\begin{align*}
I'
\le(1+p)^{r-1}\int_{1/(AN)}^{\delta'}\Big[1-\frac{A}{2}u\Big]^{N}u^{j-1}\Big\{\ln(AN)-\ln(u^{-1})+1\Big\}^{r-1}\,\dd u.
\end{align*}
The remaining calculation then proceeds as in the case $p=0$ above. The case $p<0$ is similar to the case $p>0$ and is omitted. 
\end{proof}

{
\begin{proof}[Proof of Corollary~\ref{expjoint}]
By Theorem~\ref{expdist}, we have that $\lambda T_{N,N}-b_{N}\to_{\dd}Y_{0}$ as $N\to\infty$. Thus, $S(t):=\P(\tau>t)$ is in the so-called domain of attraction \cite{haanbook} of the extreme value distribution $G_{0}(x)=\exp(-e^{-x})$ (i.e.\ \eqref{want2} holds). The desired conclusion of Corollary~\ref{expjoint} then follows from a direct application of Theorem 2.1.1 in \cite{haanbook}.
\end{proof}
}

\begin{proof}[Proof of Theorem~\ref{work}]
Since we can always rescale time, we set $\lambda=1$ without loss of generality. Note that
\begin{align}\label{niceg}
S(t+\ln A)
=Ae^{-t-\ln A}+h(t+\ln A)
=e^{-t}+g(t),
\end{align}
where we have defined $h(t):=S(t)-Ae^{-t}$ and $g(t):=h(t+\ln A)$. By assumption, $h(t)=O(e^{-\beta t})$ as $t\to\infty$, and thus
\begin{align}\label{gg}
g(t)=O(e^{-\beta t})\quad\text{as }t\to\infty.
\end{align}

Now,
\begin{align}\label{decomp}
\begin{split}
\E[(T_{N-k,N}-\ln A)^{m}]
&=\E[(T_{N-k,N}-\ln A)^{m}1_{T_{N-k,N}\ge\ln A}]\\
&\quad+\E[(T_{N-k,N}-\ln A)^{m}1_{T_{N-k,N}<\ln A}],
\end{split}
\end{align}
where $1_{E}$ denotes the indicator function on an event $A$ (i.e.\ $1_{E}=1$ is $E$ happens and $1_{E}=0$ otherwise). Now,
\begin{align*}
\big|\E[(T_{N-k,N}-\ln A)^{m}1_{T_{N-k,N}<\ln A}]\big|
\le(2\ln A)^{m}\P(T_{N-k,N}<\ln A).
\end{align*}
The assumption in \eqref{assump} implies $S(\ln A)>0$, and thus \eqref{unwieldy} implies $\P(T_{N-k,N}<\ln A)$ vanishes exponentially fast as $N\to\infty$. We thus turn our attention to the first term in \eqref{decomp}.

For any nonnegative random variable $Z\ge0$, we have that
\begin{align}\label{gnn}
\E[Z]
=\int_{0}^{\infty}\P(Z>z)\,\dd z.
\end{align}
Hence,
\begin{align*}
&\E[(T_{N-k,N}-\ln A)^{m}1_{T_{N-k,N}\ge\ln A}]\\
&=\int_{0}^{\infty}\Big[1-\P(T_{N-k,N}\le z^{1/m}+\ln A)\Big]\,\dd z\\
&=\int_{0}^{\infty}\Big[1-\sum_{i=0}^{k}{N\choose i}(1-S(z^{1/m}+\ln A))^{N-i}(S(z^{1/m}+\ln A))^{i}\Big]\,\dd z\\
&=\int_{0}^{\infty}\Big[1-\sum_{i=0}^{k}{N\choose i}\big(1-e^{-t}-g(t)\big)^{N-i}\big(e^{-t}+g(t)\big)^{i}\Big]mt^{m-1}\,\dd t,
\end{align*}
where we have used \eqref{unwieldy} and \eqref{niceg} and changed variables $t=z^{1/m}$.

Now, if $\tau_{1}',\dots,\tau_{N}'$ are iid unit rate exponential random variables and $T_{N-k,N}'$ is as in \eqref{Tkn} with $\tau_{n}'$ replacing $\tau_{n}$, then Renyi's representation \cite{renyi1953} implies 
\begin{align*}
\E[(T_{N-k,N}')^{m}]
&=\int_{0}^{\infty}\Big[1-\sum_{i=0}^{k}{N\choose i}\big(1-e^{-t}\big)^{N-i} e^{-it} \Big]mt^{m-1}\,\dd t\\
&=\E\bigg[\Big(\sum_{i=1}^{N-k}\frac{E_{i}}{i+k}\Big)^{m}\bigg],
\end{align*}
where $E_{1},\dots,E_{N-k}$ are iid unit rate exponential random variables. Hence,
\begin{align*}
&\E[(T_{N-k,N}-\ln A)^{m}1_{T_{N-k,N}\ge\ln A}]
-\E\bigg[\Big(\sum_{i=1}^{N-k}\frac{E_{i}}{i+k}\Big)^{m}\bigg]\\
&\quad=\sum_{i=0}^{k}{N\choose i}\int_{0}^{\infty}\Big[\big(1-e^{-t}\big)^{N-i} e^{-it} -\big(1-e^{-t}-g(t)\big)^{N-i}\big(e^{-t}-g(t)\big)^{i}\Big]mt^{m-1}\,\dd t\\
&\quad=\sum_{i=0}^{k}{N\choose i}\int_{0}^{\infty}\big(1-e^{-t}\big)^{N-i} e^{-it} \Big[1-\Big(\frac{1-e^{-t}-g(t)}{1-e^{-t}}\Big)^{N-i}\Big(\frac{e^{-t}-g(t)}{e^{-t}}\Big)^{i}\Big]mt^{m-1}\,\dd t\\
&\quad=\sum_{i=0}^{k}{N\choose i}\int_{0}^{\infty}\big(1-e^{-t}\big)^{N-i} e^{-it} \Big[1-e^{(N-i)[\ln(1-e^{-t}-g(t))-\ln(1-e^{-t})]}\big(1-e^{t}g(t)\big)^{i}\Big]mt^{m-1}\,\dd t.
\end{align*}

Taylor expanding and using \eqref{gg} yields that for sufficiently large $t>0$,
\begin{align*}
-2|g(t)|
\le\ln(1-e^{-t}-g(t))-\ln(1-e^{-t})
\le2|g(t)|.
\end{align*}
Hence, we may take $\delta>0$ sufficiently small so that for any $i\in\{0,1,\dots,k\}$,
\begin{align*}
&\int_{1/\delta}^{\infty}(1-e^{-t})^{N-i}(e^{-t})^{i}\Big[1-e^{2(N-i)|g(t)|}(1-e^{t}g(t))^{i}\Big]mt^{m-1}\,\dd t\\
&\quad\le\int_{1/\delta}^{\infty}(1-e^{-t})^{N-i}(e^{-t})^{i}\Big[1-e^{(N-i)[\ln(1-e^{-t}-g(t))-\ln(1-e^{-t})]}(1-e^{t}g(t))^{i}\Big]mt^{m-1}\,\dd t\\
&\quad\le\int_{1/\delta}^{\infty}(1-e^{-t})^{N-i}(e^{-t})^{i}\Big[1-e^{-2(N-i)|g(t)|}(1-e^{t}g(t))^{i}\Big]mt^{m-1}\,\dd t.
\end{align*}
Furthermore, by \eqref{gg}, we may take $\delta>0$ sufficiently small so that there exists $C>0$ so that
\begin{align*}
&\int_{1/\delta}^{\infty}(1-e^{-t})^{N-i}(e^{-t})^{i}\Big[1-e^{2(N-i)Ce^{-\beta t}}(1+Ce^{-(\beta-1)t})^{i}\Big]mt^{m-1}\,\dd t\\
&\quad\le\int_{1/\delta}^{\infty}(1-e^{-t})^{N-i}(e^{-t})^{i}\Big[1-e^{(N-i)[\ln(1-e^{-t}-g(t))-\ln(1-e^{-t})]}(1-e^{t}g(t))^{i}\Big]mt^{m-1}\,\dd t\\
&\quad\le\int_{1/\delta}^{\infty}(1-e^{-t})^{N-i}(e^{-t})^{i}\Big[1-e^{-2(N-i)Ce^{-\beta t}}(1-Ce^{-(\beta-1)t})^{i}\Big]mt^{m-1}\,\dd t.
\end{align*}
The result follows from Lemma~\ref{key2} below.
\end{proof}


\begin{lemma}\label{key2}
If $\delta>0$, $C>0$, $i\in\{0,1,2,\dots\}$, $m>0$, and $\beta>1$, then
\begin{align*}
&\int_{1/\delta}^{\infty}(1-e^{-t})^{N-i}(e^{-t})^{i}\Big[1-e^{\pm 2(N-i)Ce^{-\beta t}}(1\pm Ce^{-(\beta-1)t})^{i}\Big]t^{m-1}\,\dd t\\
&\quad=O(N^{-(\beta-1+i)}(\ln N)^{m-1})\quad\text{as }N\to\infty.
\end{align*}
\end{lemma}

\begin{proof}[Proof of Lemma~\ref{key2}]
Since
\begin{align*}
-2x
\le\ln(1-x)
\le-x\quad\text{if }x\in[0,1/2],
\end{align*}
we have that for sufficiently large $t>0$,
\begin{align*}
-2e^{-t}
\le\ln(1-e^{-t})
\le-e^{-t}.
\end{align*}
Therefore, we may take $\delta>0$ sufficiently small so that
\begin{align*}
&\int_{1/\delta}^{\infty}e^{-2(N-i)e^{-t}}(e^{-t})^{i}\Big[1-e^{- 2(N-i)Ce^{-\beta t}}(1- Ce^{-(\beta-1)t})^{i}\Big]t^{m-1}\,\dd t\\
&\quad\le\int_{1/\delta}^{\infty}e^{(N-i)\ln(1-e^{-t})}(e^{-t})^{i}\Big[1-e^{- 2(N-i)Ce^{-\beta t}}(1- Ce^{-(\beta-1)t})^{i}\Big]t^{m-1}\,\dd t\\
&\quad\le\int_{1/\delta}^{\infty}e^{-(N-i)e^{-t}}(e^{-t})^{i}\Big[1-e^{- 2(N-i)Ce^{-\beta t}}(1- Ce^{-(\beta-1)t})^{i}\Big]t^{m-1}\,\dd t,
\end{align*}
and
\begin{align*}
&\int_{1/\delta}^{\infty}e^{-(N-i)e^{-t}}(e^{-t})^{i}\Big[1-e^{ 2(N-i)Ce^{-\beta t}}(1+ Ce^{-(\beta-1)t})^{i}\Big]t^{m-1}\,\dd t\\
&\quad\le\int_{1/\delta}^{\infty}e^{(N-i)\ln(1-e^{-t})}(e^{-t})^{i}\Big[1-e^{ 2(N-i)Ce^{-\beta t}}(1+ Ce^{-(\beta-1)t})^{i}\Big]t^{m-1}\,\dd t\\
&\quad\le\int_{1/\delta}^{\infty}e^{-2(N-i)e^{-t}}(e^{-t})^{i}\Big[1-e^{ 2(N-i)Ce^{-\beta t}}(1+ Ce^{-(\beta-1)t})^{i}\Big]t^{m-1}\,\dd t.
\end{align*}
For $B>0$, changing variables $u=e^{-t}$ yields
\begin{align}
&\int_{1/\delta}^{\infty}e^{-B(N-i)e^{-t}}(e^{-t})^{i}\Big[1-e^{\pm2(N-i)Ce^{-\beta t}}(1\pm Ce^{-(\beta-1)t})^{i}\Big]t^{m-1}\,\dd t\nonumber\\
&=\int_{0}^{\delta'}e^{-B(N-i)u}u^{i-1}(\ln(1/u))^{m-1}\Big[1-e^{\pm2(N-i)Cu^{\beta}}(1\pm Cu^{\beta-1})^{i}\Big]\,\dd u,\label{second44}
\end{align}
where $\delta'=e^{-1/\delta}$. 
Expanding $e^{\pm2(N-i)Cu^{\beta}}$ about $u=0$ and applying the binomial theorem to $(1\pm Cu^{\beta-1})^{i}$yields that \eqref{second44} is the following sum of two integrals,
\begin{align}
&\int_{0}^{\delta'}e^{-B(N-i)u}u^{i-1}(\ln(u^{-1}))^{m-1}\sum_{j=1}^{i}{i\choose j}(\pm Cu^{\beta-1})^{j}\,\dd u\label{firsti}\\
&+\int_{0}^{\delta'}e^{-BN'u}u^{i-1}(\ln(u^{-1}))^{m-1}\Big(\sum_{l\ge1}\frac{(\pm2N'Cu^{\beta})^{l}}{l!}\Big)\Big(\sum_{j=0}^{i}{i\choose j}(\pm Cu^{\beta-1})^{j}\Big)\,\dd u,\label{secondi}
\end{align}
where $N'=N-i$.

To determine the large $N$ behavior of the integral in \eqref{firsti}, let $f_{0}(u)$ denote its integrand and notice that
\begin{align}\label{logsing1}
f_{0}(u)
\sim i(\pm C)u^{i-1+\beta-1}(\ln(u^{-1}))^{m-1}\quad\text{as }u\to0+.
\end{align}
We thus apply Theorem~5 in \cite{bleistein1977}, which generalizes Watson's lemma to functions with logarithmic singularities of the form \eqref{logsing1}, to obtain that the integral in \eqref{firsti}, call it $I_{0}$, satisfies
\begin{align*}
I_{0}
=O(N^{-(\beta+i-1)}(\ln N)^{m-1})\quad\text{as }N\to\infty.
\end{align*}

To determine the large $N$ behavior of the integral in \eqref{secondi}, note that
\begin{align*}
&\int_{0}^{\delta'}\sum_{l\ge1}\bigg|e^{-BN'u}u^{i-1}(\ln(1/u))^{m-1}\frac{(\pm2N'Cu^{\beta})^{l}}{l!}\Big(\sum_{j=0}^{i}{i\choose j}(\pm Cu^{\beta-1})^{j}\Big)\bigg|\,\dd u\\
&\quad=
\int_{0}^{\delta'}\sum_{l\ge1}e^{-BN'u}u^{i-1}(\ln(1/u))^{m-1}\frac{(2N'Cu^{\beta})^{l}}{l!}\Big(\sum_{j=0}^{i}{i\choose j}(\pm Cu^{\beta-1})^{j}\Big)\,\dd u\\
&\quad=
\int_{0}^{\delta'}e^{-BN'u}u^{i-1}(\ln(1/u))^{m-1}\Big(e^{2N'Cu^{\beta}}-1\Big)\Big(\sum_{j=0}^{i}{i\choose j}(\pm Cu^{\beta-1})^{j}\Big)\,\dd u<\infty
\end{align*}
assuming $\delta'=e^{-1/\delta}>0$ is sufficiently small so that
\begin{align*}
\sum_{j=0}^{i}{i\choose j}(\pm Cu^{\beta-1})^{j}>0\quad\text{for all }u\in[0,\delta'].
\end{align*}

Therefore, the Fubini-Tonelli theorem implies that
\begin{align*}
&\int_{0}^{e^{-\delta}}e^{-BN'u}u^{i-1}(\ln(1/u))^{m-1}\sum_{l\ge1}\frac{(\pm2N'Cu^{\beta})^{l}}{l!}\Big(\sum_{j=0}^{i}{i\choose j}(\pm Cu^{\beta-1})^{j}\Big)\,\dd u\\
&\quad=\sum_{l\ge1}I_{l}
\end{align*}
where
\begin{align*}
I_{l}
:=\int_{0}^{e^{-\delta}}e^{-BN'u}u^{i-1}(\ln(1/u))^{m-1}\frac{(\pm2N'Cu^{\beta})^{l}}{l!}\Big(\sum_{j=0}^{i}{i\choose j}(\pm Cu^{\beta-1})^{j}\Big)\,\dd u.
\end{align*}
To determine the behavior of $I_{l}$ as $N\to\infty$, let $f_{l}(u)$ denote the integrand of $I_{l}$ and note that it has the following singular behavior,
\begin{align}\label{logsing}
f(u)
\sim u^{i-1+\beta l}(\ln(1/u))^{m-1}\frac{(\pm2N'C)^{l}}{l!}\quad\text{as }u\to0+.
\end{align}
We thus apply Theorem~5 in \cite{bleistein1977}, which generalizes Watson's lemma to functions with logarithmic singularities of the form \eqref{logsing}, to obtain
\begin{align*}
I_{l}
=O(N^{-((\beta-1)l+i)}(\ln N)^{m-1})\quad\text{as }N\to\infty,
\end{align*}
which completes the proof.
\end{proof}

\begin{proof}[Proof of Theorem~\ref{powerdist}]
{
Let $S_{0}(t):=(\lambda t)^{-p}$ and fix $y\in\R$. Hence,
\begin{align}\label{want14}
\lim_{N\to\infty}(1-S_{0}(\lambda^{-1}N^{1/p}y))^{N}
=\lim_{N\to\infty}(1-y^{-p}/N)^{N}
=\exp(-y^{-p}).
\end{align}
By assumption, $S(t)\sim S_{0}(t)$ as $t\to\infty$, and therefore $S(\lambda^{-1}(N^{1/p}y))\sim S_{0}(\lambda^{-1}(N^{1/p}y))$ as $N\to\infty$. Hence, \eqref{want14} and Lemma~\ref{simplelemma} imply
\begin{align}\label{want24}
\lim_{N\to\infty}(1-S(\lambda^{-1}N^{1/p}y))^{N}
=\exp(-y^{-p}),\quad \text{for all }{y}\in\R.
\end{align}
Therefore, since $T_{N,N}=\max\{\tau_{1},\dots,\tau_{N}\}$ and $\{\tau_{n}\}_{n\ge1}$ are iid, we have
\begin{align}\label{triv4}
\begin{split}
\P(\lambda N^{-1/p}T_{N,N}\le y)
=\big(\P(\tau_{1}\le \lambda^{-1}N^{1/p}y)\big)^{N}
&=(1-S(\lambda^{-1}N^{1/p}y))^{N}.
\end{split}
\end{align}
Taking $N\to\infty$ in \eqref{triv4} and using \eqref{want24} yields that $\lambda N^{-1/p}T_{N,N}\to_{\dd} Y_{0}$, which completes the proof for the case $k=0$. Having established the case $k=0$, the case of a general fixed integer $k\ge1$ and the convergence in distribution for the joint random variables in \eqref{cd4} follows from Theorem 2.1.1 in \cite{haanbook}.
}
\end{proof}


\begin{proof}[Proof of Theorem~\ref{powermoments}]
Since Theorem~\ref{powerdist} establishes that $\lambda T_{N-k,N}/N^{1/p}$ converges in distribution to $Z_{k}$ as $N\to\infty$, the continuous mapping theorem (see, for example, Theorem~2.7 in \cite{billingsley2013}) implies that $(\lambda T_{N-k,N}/N^{1/p})^{m}$ converges in distribution to $(Z_{k})^{m}$ as $N\to\infty$. To conclude $\E[(\lambda T_{N-k,N}/N^{1/p})^{m}]\to\E[(Z_{k})^{m}]$ as $N\to\infty$, it is enough to show that the random variables $\{(\lambda T_{N-k,N}/N^{1/p})^{m}\}_{N\ge1}$ are uniformly integrable (see, for example, Theorem~3.5 in \cite{billingsley2013}). To show this uniform integrability, it is enough (see, for example, equation~(3.18) in \cite{billingsley2013}) to show that
\begin{align}\label{wantui}
\sup_{N\ge1}\E\Big[\big(\lambda T_{N-k,N}/N^{1/p}\big)^{r}\Big]<\infty,
\end{align}
where $r$ is an even integer satisfying $r>m>0$.

Since $S(t)\sim(\lambda t)^{-p}$ as $t\to\infty$, there exists $t_{0}\ge1/\lambda$ such that
\begin{align*}
S(t)\le2(\lambda t)^{-p}<1\quad\text{for all }t\ge t_{0}\ge1/\lambda.
\end{align*}
Suppose $\{\tau_{n}^{+}\}_{n\ge1}$ is an iid sequence of realizations of a random variable with survival probability
\begin{align*}
S_{+}(t)
=\begin{cases}
1 & \text{if }t<t_{0},\\
(2\lambda t)^{-p} & \text{if }t\ge t_{0}.
\end{cases}
\end{align*}
Defining $T_{N-k,N}^{+}$ as in \eqref{Tkn} but with $\tau_{n}$ replaced by $\tau_{n}^{+}$, it is immediate that
\begin{align}\label{ebp}
\E[(T_{N-k,N})^{r}]
\le\E[(T_{N-k,N}^{+})^{r}].
\end{align}
Using \eqref{gnn} and \eqref{unwieldy}, we then have
\begin{align*}
\E[(T_{N-k,N}^{+})^{r}]
&=\int_{0}^{\infty}\P(T_{N-k,N}^{+}>s^{1/r})\,\dd s\\
&=\sum_{i=0}^{N-k-1}{N\choose i}\int_{0}^{\infty}(S_{+}(s^{1/r}))^{N-i}(1-S_{+}(s^{1/r}))^{i}\,\dd s.
\end{align*}
For the $i=0$ term, we have that
\begin{align*}
\int_{0}^{\infty}(S_{+}(s^{1/r}))^{N}\,\dd s
=(t_{0})^{r}+\int_{(t_{0})^{r}}^{\infty}(2\lambda s^{1/r})^{-Np}\,\dd s
=(t_{0})^r \Big(\frac{r 2^{-N p} (\lambda  t_{0})^{-N p}}{N p-r}+1\Big).
\end{align*}
For $i\in\{1,\dots,N-k-1\}$, we have
\begin{align*}
&\int_{0}^{\infty}(S_{+}(s^{1/r}))^{N-i}(1-S_{+}(s^{1/r}))^{i}\,\dd s\\
&\quad=\int_{(t_{0})^{r}}^{\infty}(2\lambda s^{1/r})^{-p(N-i)}(1-(2\lambda s^{1/r})^{-p})^{i}\,\dd s\\
&\quad=(2\lambda)^{-r}\int_{(2\lambda t_{0})^{r}}^{\infty}(z^{1/r})^{-p(N-i)}(1-(z^{1/r})^{-p})^{i}\,\dd z\\
&\quad\le(2\lambda)^{-r}\int_{1}^{\infty}(z^{1/r})^{-p(N-i)}(1-(z^{1/r})^{-p})^{i}\,\dd z,
\end{align*}
where we have used that $\lambda t_{0}\ge1$. Changing variables $u=z^{-r/p}$ yields
\begin{align*}
\int_{1}^{\infty}(z^{1/r})^{-p(N-i)}(1-(z^{1/r})^{-p})^{i}\,\dd z
&=\frac{r}{p}\int_{0}^{1}u^{-1-r/p+N-i}(1-u)^{i}\,\dd u\\
&=\frac{r}{p}\frac{\Gamma(N-i-r/p)\Gamma(i+1)}{\Gamma(N-r/p+1)},
\end{align*}
where we have used the identity,
\begin{align*}
\int_{0}^{1}u^{x-1}(1-u)^{y-1}\,\dd u
=\frac{\Gamma(x)\Gamma(y)}{\Gamma(x+y)}.
\end{align*}
Now, using the identity $\Gamma(z+1)=z\Gamma(z)$, one can check that
\begin{align*}
\sum_{i=0}^{N-k-1}\frac{\Gamma(N-i-r/p)}{\Gamma(N-i+1)}
=\frac{p}{r}\Big(\frac{\Gamma(k+1-r/p)}{\Gamma(k+1)}-\frac{\Gamma(N+1-r/p)}{\Gamma(N+1)}\Big).
\end{align*}
Therefore,
\begin{align*}
\sum_{i=0}^{N-k-1}{N\choose i}\frac{r}{p}\frac{\Gamma(N-i-r/p)\Gamma(i+1)}{\Gamma(N-r/p+1)}
=\frac{\binom{N}{k} \Gamma (N-k+1) \Gamma (k-\frac{r}{p}+1)}{\Gamma(N-\frac{r}{p}+1)}-1.
\end{align*}
Taking $N\to\infty$ yields
\begin{align*}
\lim_{N\to\infty}N^{-r/p}\Big(\frac{\binom{N}{k} \Gamma (N-k+1) \Gamma (k-\frac{r}{p}+1)}{\Gamma(N-\frac{r}{p}+1)}-1\Big)=\frac{\Gamma(k+1-r/p)}{\Gamma(k+1)}<\infty.
\end{align*}
Combining this calculation with \eqref{ebp} yields
\begin{align*}
0
\le\limsup_{N\to\infty}N^{-r/p}\E[(T_{N-k,N})^{r}]
\le\limsup_{N\to\infty}N^{-r/p}\E[(T_{N-k,N}^{+})^{r}]<\infty,
\end{align*}
and thus \eqref{wantui} holds.
\end{proof}

\subsection{Numerical methods and auxiliary proofs}\label{appnumerical}

We now give more details on the numerical methods used in sections~\ref{meno}-\ref{secdiscussion}.

\subsubsection{Numerical methods for section~\ref{meno}}\label{appmeno}

The cumulative distribution function in \eqref{cdfsim} was obtained by $M=10^{3}$ stochastic realizations of $T_{N-k,N}$ using the survival probability of $\tau$ in \eqref{sre}. In particular, $M\times N$ independent realizations of $\tau$ were sampled, which then yielded $M$ independent realizations of $T_{N-k,N}$. Each realization of $\tau$ was obtained by numerically inverting $S(\tau)=U$, where $U$ is a uniformly distributed random value on $[0,1]$ and $S(t)$ is computed from the first 100 terms in the series in \eqref{sre}.

\subsubsection{Numerical methods for section~\ref{interval}}\label{appinterval}

For the example in section~\ref{interval}, we compute the statistics and distribution of $T_{N,N}$ using the first 100 terms in the following series representation for $S(t)$,
\begin{align}\label{serinterval} 
\begin{split}
S(t)
&=\sum_{k=1}^{\infty}A_{k}e^{-\lambda_{k}t},\\
\lambda_{k}
&=Dk^{2}\pi^{2}/(2L)^{2},\quad
A_{k}=\frac{4}{k\pi}\sin(k\pi(x_{0}+L)/(2L)).
\end{split}
\end{align}
The series in \eqref{serinterval} is obtained by finding the solution $s(x,t)$ to the backward Kolmogorov equation \cite{pavliotis2014}
\begin{align*}
\partial_{t}s
&=D\partial_{xx}s,\quad x\in(0,2L),
\end{align*}
with $s=0$ if $x\in\{0,2L\}$ and $s=1$ if $t=0$ and setting $S(t)=s(x_{0}+L,t)$. {Generalizing the case that each searcher starts at $x_{0}\in(-L,L)$, if we instead assume that each searcher starts according to some stochastic initial position with probability measure $\mu$, then we merely set
\begin{align*}
S(t)=\int_{-L}^{L}s(x_{0}+L,t)\,\dd \mu(x_{0}),
\end{align*}
which merely amounts to replacing $A_{k}$ in \eqref{serinterval} by $\int_{-L}^{L}A_{k}\,\dd \mu(x_{0})$.} The first and second moments of $T_{N,N}$ were computed via the following integrals using the trapezoidal rule,
\begin{align}\label{ttrap}
\E[T_{N,N}]
&=\int_{0}^{\infty}(1-(1-S(t))^{N})\,\dd t,\\
\E[(T_{N,N})^{2}]
&=\int_{0}^{\infty}2t(1-(1-S(t))^{N})\,\dd t,\label{ttrap2}
\end{align}
To compute \eqref{ttrap}-\eqref{ttrap2}, we use $10^{7}$ uniformly spaced time points from time $t=0$ to time $t=10$.

\subsubsection{Numerical methods for section~\ref{rare}}\label{apprare}

For the example in section~\ref{rare}, we find the solution $s(r,t)$ to the backward Kolmogorov equation \cite{pavliotis2014}
\begin{align*}
\partial_{t}s
&=D(\partial_{rr}s+(2/r)\partial_{r}s),\quad r\in(a,L),
\end{align*}
with $s=0$ if $r=a$, $\partial_{r}s=0$ if $r=L$, and $s=1$ if $t=0$ and setting $S(t)=s(L,t)$. We find $s(r,t)$ numerically using \texttt{pdepe} in Matlab \cite{MATLAB} with $10^{5}$ equally uniform spatial grid points between $r=a$ and $r=L$ and $10^{4}$ logarithmically spaced time points between $t=10^{-16}$ and $t=10^{4}/3$. We then calculate $\E[T_{N,N}]$ via \eqref{ttrap} using the trapezoidal rule on these time points.

\subsubsection{Details for section~\ref{sg}}\label{appkappa}

For the example in Figure~\ref{figkappa} in section~\ref{sg}, we calculate $\E[T_{N,N}]$ via computing the integral in \eqref{ttrap} with the trapezoidal rule, where $S(t)$ is computed by taking the first $10^{3}$ terms in the series representation given in equation~(B11) in \cite{grebenkov2022}. To compute \eqref{ttrap}, we use $10^{6}$ uniformly spaced time points from time $t=0$ to time $t=10^{3}$

\subsubsection{Details for section~\ref{halfline}}\label{apphalfline}

For the numerics in section~\ref{halfline}, we use the exact formula for $S(t)$ in \eqref{erf}. The formula for $S(t)=s(x,t)$ in \eqref{SV} is obtained by checking that it satisfies the backward Kolmogorov equation \cite{pavliotis2014}
\begin{align*}
\partial_{t}s
=D\partial_{xx}s-V\partial_{x}s,\quad x>0,
\end{align*}
and $s=0$ if $x=0$ and $s=1$ if $t=0$.

\subsubsection{Auxiliary proofs for section~\ref{sub}}\label{appsub}

For the example in section~\ref{sub}, we have that $S_{\sigma}(t)$ is given by \eqref{serinterval} and thus \eqref{rep} implies
\begin{align*}
S(t)
=\E[S_{\sigma}(\mathbb{S}(t))]
=\sum_{k=1}^{\infty}A_{k}\E[e^{-\lambda_{k}\mathbb{S}(t)}]
=\sum_{k=1}^{\infty}A_{k}E_{\alpha}(-\lambda_{k}t^{\alpha}),
\end{align*}
where $E_{\alpha}$ is the Mittag-Leffler function,
\begin{align*}
E_{\alpha}(z)
:=\sum_{n=0}^{\infty}\frac{z^{n}}{\Gamma(1+\alpha n)},
\end{align*}
and we have used that
\begin{align}\label{gfm}
\E[e^{-\lambda t}]
=E_{\alpha}(-\lambda t^{\alpha}),\quad\text{if }t>0,\lambda>0.
\end{align}
The distributions in the right panel of Figure~\ref{fighalfsub} use the first 100 terms of the series \eqref{gfm}. 
To obtain \eqref{gfm}, recall the probability density function of $\mathbb{S}(t)$ in \eqref{lspdf}, integrate by parts, and use the series representation for the exponential function, 
\begin{align*}
\E[e^{-\lambda t}]
=\int_{0}^{\infty}e^{-\lambda s}\frac{t}{\alpha s^{1+1/\alpha}}l_{\alpha}\Big(\frac{t}{s^{1/\alpha}}\Big)\,\dd s
&=\int_{0}^{\infty}e^{-\lambda t^{\alpha}z^{-\alpha}}l_{\alpha}(z)\,\dd z\\
&=\int_{0}^{\infty}\sum_{n=0}^{\infty}\frac{(-\lambda t^{\alpha}z^{-\alpha})^{n}}{\Gamma(n+1)}l_{\alpha}(z)\,\dd z\\
&=\sum_{n=0}^{\infty}\frac{(-\lambda t^{\alpha})^{n}}{\Gamma(n+1)}\int_{0}^{\infty}z^{-\alpha n}l_{\alpha}(z)\,\dd z\\
&=E_{\alpha}(-\lambda t^{\alpha}),
\end{align*}
where the final equality uses the following formula for moments of a one-sided Levy stable distribution \cite{penson2010},
\begin{align*}
\int_{0}^{\infty}z^{\mu}l_{\alpha}(z)
=\frac{\Gamma(-\mu/\alpha)}{\alpha\Gamma(-\mu)},\quad -\infty<\mu<\alpha.
\end{align*}

\begin{proof}[Proof of Proposition~\ref{thmsub}]
Let $\eps>0$. It is well-known that $l_{\alpha}(z)$ has the following asymptotic behavior \cite{barkai2001},
\begin{align*}
l_{\alpha}(z)
\sim\frac{\alpha}{\Gamma(1-\alpha)}z^{-1-\alpha}\quad\text{as }z\to\infty.
\end{align*}
Hence, there exists a $\delta>0$ so that
\begin{align}\label{bounds}
(1-\eps)\frac{\alpha}{\Gamma(1-\alpha)}z^{-1-\alpha}
\le l_{\alpha}(z)
\le(1+\eps)\frac{\alpha}{\Gamma(1-\alpha)}z^{-1-\alpha},\quad\text{if }z>1/\delta.
\end{align}
If we split the integral in \eqref{reprep} into two integrals,
\begin{align*}
S(t)
&=\int_{0}^{(\delta t)^{\alpha}}\P(\sigma>s)\frac{t}{\alpha s^{1+1/\alpha}}l_{\alpha}\Big(\frac{t}{s^{1/\alpha}}\Big)\,\dd s
+\int_{(\delta t)^{\alpha}}^{\infty}\P(\sigma>s)\frac{t}{\alpha s^{1+1/\alpha}}l_{\alpha}\Big(\frac{t}{s^{1/\alpha}}\Big)\,\dd s,\\
&=:I_{1}+I_{2},
\end{align*}
then \eqref{bounds} implies that we can bound the first integral, $I_{1}$, as follows,
\begin{align*}
\frac{(1-\eps)t^{-\alpha}}{\Gamma(1-\alpha)}\int_{0}^{(\delta t)^{\alpha}}\P(\sigma>s)\,\dd s
\le I_{1}
\le\frac{(1+\eps)t^{-\alpha}}{\Gamma(1-\alpha)}\int_{0}^{(\delta t)^{\alpha}}\P(\sigma>s)\,\dd s.
\end{align*}
Hence,
\begin{align*}
1-\eps
&\le\liminf_{t\to\infty}\frac{I_{1}}{\frac{t^{-\alpha}}{\Gamma(1-\alpha)}\int_{0}^{(\delta t)^{\alpha}}\P(\sigma>s)\,\dd s}\\
&\quad\le\limsup_{t\to\infty}\frac{I_{1}}{\frac{t^{-\alpha}}{\Gamma(1-\alpha)}\int_{0}^{(\delta t)^{\alpha}}\P(\sigma>s)\,\dd s}
\le 1+\eps.
\end{align*}
Since $\eps>0$ is arbitrary and since we assumed $\E[\tau]=\int_{0}^{\infty}\P(\sigma>s)\,\dd s<\infty$, we obtain
\begin{align*}
I_{1}
\sim\frac{\E[\sigma]}{\Gamma(1-\alpha)}t^{-\alpha}\quad\text{as }t\to\infty.
\end{align*}

It remains to show that the second integral, $I_{2}$, vanishes faster than $t^{-\alpha}$ as $t\to\infty$. Since $\P(\sigma>s)$ is a nonincreasing function of $s\ge0$, and $\frac{t}{\alpha s^{1+1/\alpha}}l_{\alpha}(\frac{t}{s^{1/\alpha}})$ is a probability density function, we have that
\begin{align}\label{boundI2}
I_{2}
\le \P\big(\sigma>(\delta t)^{\alpha}\big)\int_{(\delta t)^{\alpha}}^{\infty}\frac{t}{\alpha s^{1+1/\alpha}}l_{\alpha}\Big(\frac{t}{s^{1/\alpha}}\Big)\,\dd s
\le\P\big(\sigma>(\delta t)^{\alpha}\big).
\end{align}
Since we assumed $\E[\tau]=\int_{0}^{\infty}\P(\sigma>s)\,\dd s<\infty$, it follows that $\P(\sigma>s)$ must vanish faster than $s^{-1}$ as $s\to\infty$. Hence, \eqref{boundI2}  completes the proof.
\end{proof}

\subsubsection{Auxiliary proof for section~\ref{mortal}}\label{appmortal}

\begin{proof}[Proof of Proposition~\ref{mortalthm}]
Integrating over the possible values of $\sigma$ gives
\begin{align}
\P(\tau\le\sigma)
&=\int_{0}^{\infty}(1-S(s))r e^{-r s}\,\dd s,\nonumber\\
\P(t<\tau\le\sigma)
&=S(t)e^{-r t}-\int_{t}^{\infty}S(s)r e^{-r s}\,\dd s.\label{irep}
\end{align}
The results follow from the definition of conditional probability in \eqref{taubar} and applying Laplace's method to the integral in \eqref{irep}.
\end{proof}

\section*{Ethical statement}

The authors declare that there is no conflict of interest.

\section*{Data availability statement }
Data will be made available on reasonable request.

\bibliography{library.bib}
\bibliographystyle{unsrt}

\end{document}